\documentclass[preprint,3p,lefttitle]{elsarticle}

\usepackage{amssymb}
\usepackage{latexsym}

\usepackage{url}
\usepackage{xcolor}
\definecolor{newcolor}{rgb}{.8,.349,.1}

\usepackage{soul}  
\usepackage{xcolor}  
\newcommand{\bl}[1]{\textcolor{black}{#1}}  
\newcommand{\red}[1]{\textcolor{black}{#1}}

\usepackage{amsmath}
\usepackage{listings}
\usepackage{algorithmic}
\usepackage{textgreek}
\usepackage{amsthm}
\usepackage{amssymb}
\usepackage{mdframed}
\usepackage[ruled,vlined]{algorithm2e}[H]
\usepackage[bookmarks=false]{hyperref}
\theoremstyle{definition}
\newcommand{\comment}[1]{}

\newtheorem{theorem}{Theorem}

\newtheorem{proof1}{Proof}
\newtheorem{proposition}{Proposition}

\usepackage{mdframed}

\newmdtheoremenv{problem_stmt}{Problem}

\DeclareMathOperator*{\argminA}{arg\,min}
\journal{Ad Hoc Networks}

\begin{document}

\begin{frontmatter}


\title{Secure and Ultra-Reliable Provenance Recovery in Sparse Networks: Strategies and Performance Bounds}

\author[1]{Suraj Sajeev}
\author[2]{Manish Bansal}
\author[1]{Sriraam S V}
 \author[2]{Harshan Jagadeesh}
\author[1]{Huzur Saran}
\author[3]{Yih-Chun Hu}

\address[1]{Department of Computer Science and Engineering, Indian Institute of Technology Delhi, India}
\address[2]{Bharti School of Telecommunication Technology and Management, Indian Institute of Technology Delhi, India}

\address[3]{Department of Electrical and Computer Engineering, University of Illinois Urbana-Champaign, U.S.A}

\begin{abstract}
Provenance embedding algorithms are well known for tracking the footprints of information flow in wireless networks. Recently, low-latency provenance embedding algorithms have received traction in vehicular networks owing to strict deadlines on the delivery of packets. While existing low-latency provenance embedding methods focus on reducing the packet delay, they assume a complete graph on the underlying topology due to the mobility of the participating nodes. We identify that the complete graph assumption leads to sub-optimal performance in provenance recovery, especially when the vehicular network is sparse, which is usually observed outside peak-hour traffic conditions. As a result, we propose a two-part approach to design provenance embedding algorithms for sparse vehicular networks. In the first part, we propose secure and practical topology-learning strategies, whereas in the second part, we design provenance embedding algorithms that guarantee ultra-reliability by incorporating the topology knowledge at the destination during the provenance recovery process. Besides the novel idea of using topology knowledge for provenance recovery, a distinguishing feature for achieving ultra-reliability is the use of hash-chains in the packet, which trade communication-overhead of the packet with the complexity-overhead at the destination. We derive tight upper bounds on the performance of our strategies, and show that the derived bounds, when optimized with appropriate constraints, deliver design parameters that outperform existing methods. Finally, we also implement our ideas on OMNeT++ based simulation environment to show that their latency benefits indeed make them suitable for vehicular network applications.
\end{abstract}




\end{frontmatter}


\section{Introduction}
\label{sec:introduction}

With the rapid increase in use-cases for wireless communication, deployment of wireless networks for enterprise and public infrastructures are in demand. Example applications include sensor networks for acquiring real-time spatio-temporal data \cite{SPC1}, vehicular networks for facilitating urban transportation \cite{VSN}, and Industrial IoT (IIoT) for improving efficiency in industrial automation \cite{INDA}. One of the key features envisaged by these networks is the self-organizing capability wherein the wireless nodes of the network can communicate among each other or with a control center. 

Although the self-organizing capability in wireless networks provides the much needed scalability feature, it also exposes the network to various security threats on the trustworthiness of the data flowing through the network. While non-critical applications such as temperature sensing using sensor networks can address security threats through the use of long-term checks in an intermittent manner, vehicular networks that deploy such wireless networks will have to handle security threats instantaneously on a packet-to-packet basis. This is because any packet in the network is vulnerable to security threats, which in turn may lead to catastrophic consequences \cite{SVN}. A standard way to detect security threats is to use a fixed portion of the packet to carry provenance information of the packet from its origin to the destination \cite{SURV}. This way, upon receiving the packet, the control center will recover the identity of the processes that modified the packet, and then detect any security breach that might have occurred on the packet. Formally, \textit{provenance} \cite{SPS},\cite{MP} refers to the information on the origin of the data, the process that has modified the data and various nodes that have forwarded the data in the network. However, in the context of this work, provenance does not just refer to the information on the set of nodes that forwarded the packet; it also captures the order of nodes (i.e., the path) through which the packets reach the destination. As a potential use-case, provenance serves as meta-data by assisting the destination in verifying the authenticity of the packet thereby forbidding an adversary in either illegitimately modifying the data or executing a denial of service attack. For instance, in applications such as multi-hop distance bounding algorithms \cite{dist_bound1}, \cite{dist_bound2}, the destination might be interested in learning the number of hops through which the packets reach the destination. In such scenarios asking the intermediate nodes to only embed their identities, i.e., only the information on the set of nodes, in the packet may lead to denial-of-service threats as a node might not embed its signature. Consequently, the destination estimates incorrect bounds on the distance between the source node and itself. However, when the destination ask the nodes to embed signatures in the packet based on their edges with their neighbours (i.e., by capturing the order of nodes), such denial-of-service threats can be detected at the destination since a path is not formed in the provenance recovery process. This is one use-case where provenance is potentially applicable in vehicular networks. As another use-case, provenance is helpful in multi-hop communication wherein a number of nodes in the network are supposed to modify the data of the packet en-route to the destination such as statistics addition, state change notifications etc. In such a scenario, the knowledge on the order in which the nodes participated in the update process will help the destination to detect integrity threats on the packets.

While provenance recovery in wireless networks is imperative to detect and mitigate security threats, it is well known that such benefits come at the cost of additional communication-overhead as well as additional delay-overhead on the packets. With the emergence of new use-cases of wireless networks in time-critical vehicular networks,  it has been recently shown that legacy provenance recovery methods cannot be applied to vehicular networks in an off-the-shelf manner \cite{BFP}. This is because the delay-overheads offered by the legacy techniques may not meet the latency constraints on vehicular networks. As a result, it has been shown that new provenance recovery methods must be designed from first principles to cater to the mobility and low-latency requirements of vehicular networks. In particular, \cite{BFP} proposed a provenance recovery framework for vehicular networks to reduce the delays on the packets when the intermediate relay nodes modify the packets on their way to the destination. Although \cite{BFP} addressed the constraint of delay, some of the other key features, such as the network's density, have been omitted. For instance, it is well known that the density pattern of a vehicular network within a Road Side Unit (RSU) depends on several factors, such as the coverage range of RSU, and also the real-time traffic pattern, which in turn depends on the time of the day as well as the city/town wherein the vehicular network is deployed \cite{rsparse}. Inspired by these practical scenarios, we explore the problem of designing secure provenance recovery algorithms by capitalizing on the knowledge of the density pattern of the underlying vehicular network. 



\begin{table*}
\caption{\label{table:summary}Summary of contributions. As depicted in Fig. \ref{fig:comparison}, the regions B1, B2, B3, B4, B5 denote the specific areas where prior works exist in topology learning algorithms. Similarly, the regions A1, A2, A3, $\ldots$, A7 denote the specific areas where prior works exist in provenance recovery algorithms. The region A7 captures those contributions that neither address latency constraint, nor complexity constraints, nor known topology.}
\vspace{0.5cm}
\centering
\resizebox{16cm}{!}{
\begin{tabular}{|c|c|c|c|}
\hline
\textbf{Main contributions} & \textbf{Specific} & \textbf{Salient features} & \textbf{Existing contributions}\\
 & \textbf{contributions} & &\\
\hline
Topology learning & SSMP & Ultra-reliability and high communication-overhead. & \cite{BFP} $\in$ A7 assumed complete graph\\
protocols & & Derived performance bounds. & \cite{TOP1} $\in$ B1 limited to learning a \\
& & Multi-variable optimization problem. & subset of the topology.\\
 & &  Optimization challenging to solve. & \cite{TOP2} $\in$ B1 applicable to linear networks,\\
 & & & high communication-overhead\\
  & MSSP & Low reliability and low communication overhead. & \cite{TOP3} $\in$ B4, \cite{TOP4} $\in$ B3 need dedicated devices:\\
  & & Derived performance bounds. & \cite{TOP3} uses monitors, \cite{TOP4} uses UAVs\\
  & & Single-variable optimization problem. & to learn physical topology.\\
 &  & Optimization straightforward to solve. &   \cite{TOP5} $\in$ B1 designed for \\
  & &  & non-adversarial networks\\
\hline
 Provenance recovery  & DE & Uses hash-chain & \cite{BFP} did not use hash-chain\\
 methods &  & to trade-off communication-overhead & \cite{BFP} assumed complete graph\\
 & & with complexity-overhead & We outperform \cite{BFP}.\\
  & & Derived performance bounds with topology. & Our proofs are more elegant \\
   & & & than \cite{BFP}.\\
& DDE & Uses hash-chain & \cite{BFP} did not use hash-chain.\\
 &  & to trade-off communication-overhead & \cite{BFP} assumed complete graph.\\
 & & with complexity-overhead & We outperform \cite{BFP}\\
  & & Derived performance bounds with topology. & Our proofs are more elegant \\
  & & & than \cite{BFP}.\\
\hline
 OMNeT++ based & & SSMP provides low end-to-end delay & Prior works have \\
simulations & & although it has high communication-overhead. & not implemented topology \\
& & Topology knowledge reduces delay & learning and \\
& & on provenance recovery. & provenance recovery jointly.\\
\hline
\end{tabular}}
\end{table*}

\begin{figure}
\begin{center}
\includegraphics[trim={0 0 0 0},clip,scale = 0.45]{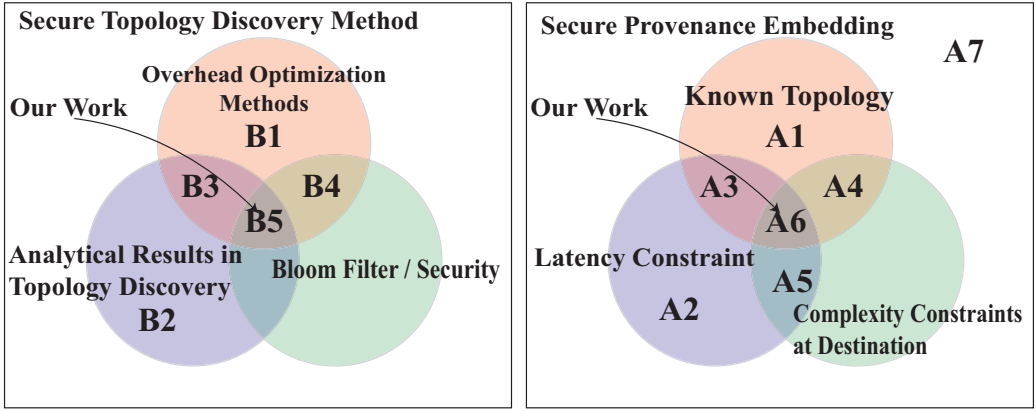}
\end{center}
\caption{An overview of our contributions w.r.t. existing contributions. \cite{TOP1}, \cite{TOP2}, \cite{TOP5} $\in$ B1; \cite{B31} $\in$ B2; \cite{TOP4}, \cite{TOP6} $\in$ B3; \cite{TOP3} $\in$ B4; \cite{SP2} $\in$ A1; \cite{PML}, \cite{LDP} $\in$ A2;  \cite{LSS} $\in$ A3; \cite{PDA} $\in$ A4; \cite{LWSA}, \cite{PAP} $\in$ A5; \cite{BFP}, \cite{SPS}, \cite{MP}, \cite{SP1} $\in$ A7.}
\label{fig:comparison}
\end{figure} 


\subsection{Motivation and Contributions}

Provenance recovery with applications to vehicular networks throws unique challenges due to the mobility of the underlying nodes. To handle this mobility constraint, existing methods for provenance recovery have taken up a conservative approach of assuming a complete graph \cite{WAXM} on the underlying topology, i.e., the assumption that every node is potentially connected to every other node. As a significant departure from this approach, we point out that the underlying network topology of vehicles may be sparse, i.e., the total number of edges can be far fewer than that of a complete graph. In particular, \cite{rsparse}, \cite{RSUdeploy} (and references within) have shown that sparse vehicular networks are observed in real-time traffic reports, especially outside peak-hour traffic, and moreover, a wide number of networking protocols have been proposed to particularly cater to sparse vehicular networks. As a result of these developments in sparse vehicular networks, we point out that the parameters of the provenance recovery methods that are chosen based on the assumption of a complete graph do not necessarily match the topology, thereby leading to sub-optimal performance. Motivated by this observation, we introduce a new provenance framework wherein we first learn the topology of the network and then exploit its knowledge, i.e., the sparsity of the underlying graph, to design provenance embedding algorithms. Subsequently, we address various objectives such as reducing the error rates in provenance recovery, handle latency constraints (required for vehicular networks), and also tackle various adversarial attacks. The contributions of our work, which can be divided into two phases, namely: (i) the topology learning phase and (ii) the payload phase, are discussed below.

1) For the topology learning phase, we propose two Bloom filter based algorithms, namely: (i) Single-Source Multi-Packet (SSMP) embedding, and (ii) Multi-Source Single-Packet (MSSP) embedding (see Section \ref{sec:Toplogy Learning Protocols}). The two schemes are such that each node uses pre-shared signatures with the destination to embed the information on its neighbours into the Bloom filter, and then the destination, after recovering the Bloom filter portions from the packet(s), constructs the topology. A salient feature of our approach is that an edge is considered a member of the topology only if both its vertices embed the edge in the Bloom filter. For both the learning schemes, we study the behaviour of false-positive events, wherein non-existing edges are also retrieved as valid edges of the topology owing to the use of Bloom filters. To choose the right set of parameters for our algorithms, we derive closed-form expression on the false-positive rates (FPRs) and their upper-bounds as a function of Bloom filter parameters and the information on the number of neighbours of each node (see Theorem \ref{thm:SSMP_exact}, Theorem \ref{thm:SSMP_bound} and Theorem \ref{thm:MSSP}). Through extensive simulation results, we show that our performance bounds help us for parameter estimation in practice. Other than studying the FPRs of the proposed methods, we also study various aspects such as transmission overhead, resilience to threats, and ease of optimization for parameter estimation. Unlike the existing contributions in this space, as depicted on the left-side of Fig. \ref{fig:comparison}, our work addresses joint optimization of Bloom filter sizes of all the nodes, and also presents analytical results that assist in parameter estimation for the discovery algorithms. As listed in Table \ref{table:summary}, existing contributions have only focused on routing overheads in non-adversarial environments without incorporating the packet size for optimization. 
   
2) In the payload phase, we design low-latency provenance embedding and provenance recovery algorithms that use the topology knowledge gathered during the topology learning phase (see Section \ref{sec:payload}). We propose hash-chain assisted variants of (i) deterministic edge embedding (DE), and (ii) deterministic double-edge embedding (DDE) methods that were recently proposed to handle the mobility and low-latency feature of vehicular network \cite{BFP}. Although the DE and DDE methods were proposed to handle unknown topology in \cite{BFP}, we observe that these methods are still applicable with the topology knowledge as they help to resolve multiple paths in topologies containing cycles. Furthermore, the DDE method continues to assist low-latency communication by asking the intermediate nodes to skip the embedding process thereby reducing the end-to-end delay compared to the DE method. As the main contribution of this part of our work, we use hash-chains to entangle the edge identities (or the double-edge identities) of successive nodes so that the destination first generates a list of candidate paths from the received Bloom filter, and then resolves the exact path by recomputing the hash-chains of all the candidate paths. This way, we study the trade-off between the communication overhead captured through the Bloom filter size in the packet and the complexity-overhead captured through the affordable number of hash-chain computations at the destination. In particular, we consider hash-chain assisted DE and DDE methods wherein a Bloom filter of size $m$ bits is used in the packet to carry the provenance information, and the destination is allotted a complexity of verifying up to $\beta$ candidate paths, for some $\beta > 1$, to recover the provenance. For these settings, assuming perfect topology knowledge at the destination, we analyze the FPRs of the proposed provenance recovery methods in sparse vehicular networks, and then derive upper bounds on them as a function of $m$ and $\beta$ (see Theorem \ref{thm_DEE_PFA}). We show that the derived performance bounds help us to choose the right values of $m$ and $\beta$ for a given FPR. We also present extensive simulation results to show that the parameters obtained using the performance bounds are approximately the same as those derived using simulation results. Overall, when compared to \cite{BFP}, we show that topology knowledge significantly improves the accuracy of provenance recovery, especially when the topology is sparse. As depicted in Fig. \ref{fig:comparison}, this is the first work that addresses low-latency provenance recovery with topology knowledge for various complexity constraints at the destination.

3) To demonstrate the impact of our framework, we implement the proposed topology learning phase and the provenance recovery phase on OMNeT++ environment. Based on these experiments, we observe that the SSMP algorithm achieves substantial end-to-end delay reduction compared to the MSSP algorithm, and this behaviour is attributed to the fact that the edge verification procedure at the RSU can be initiated as when packets from individual nodes arrive, while the packets from other nodes are en-route to the RSU. In contrast, the MSSP algorithm requires the RSU to initiate the edge verification procedure only after the packet arrives after traversing all the nodes in the network. During the packet routing process of the SSMP algorithm, packets from different nodes can be routed simultaneously as long as queuing is appropriately handled at each node. In contrast, in the MSSP algorithm, this parallelism is not possible since every node has to wait to receive the packet sequentially, and then embed its neighbours in the Bloom filter. Based on the simulation results on the payload phase, we observe that with the topology knowledge for the provenance recovery method, the end-to-end delay achieved in the provenance recovery step at the RSU is significantly lower compared to that when complete graph is used. Although this result is intuitive, our simulation results on OMNeT++ makes this observation explicit, thereby driving home the point that the knowledge of the topology reduces the latency in the provenance recovery process in addition to the delay benefits of the DDE method during packet routing.

\begin{figure}
\begin{center}
\includegraphics[trim={0 4cm 0cm 0cm},clip,scale = 0.4]{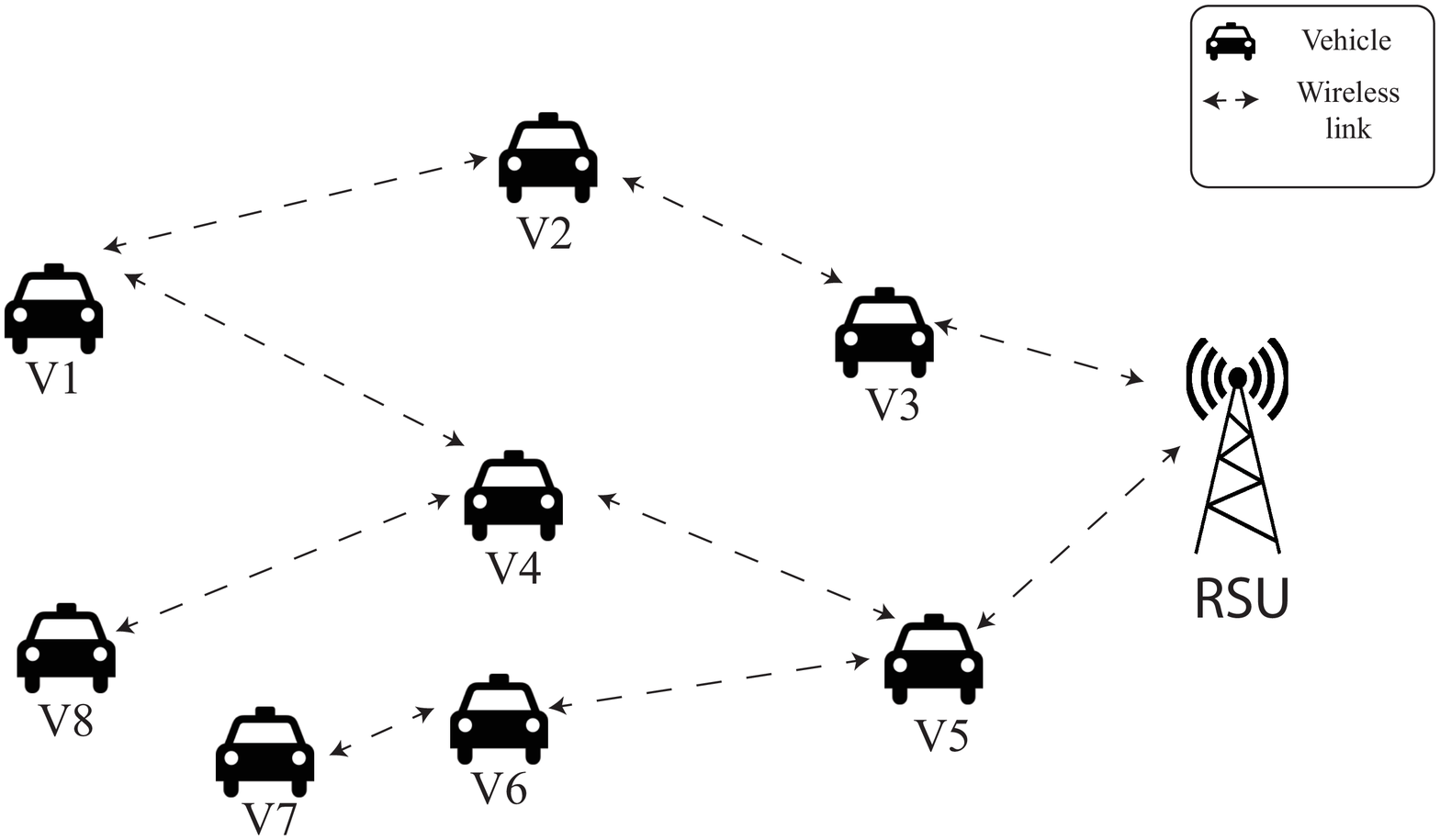}
\end{center}
\caption{Figure depicting the logical topology of the vehicular network, when two vehicles are not connected, it means they are out of coverage area of each other.  With no prior knowledge on the network topology, the destination must first learn the topology of the network, and then determine the path traced by the packets in the payload phase.}
\label{fig:network}
\end{figure}

\subsection{Related Work on Network Provenance using Bloom Filters}
\label{sec:bf_basics}

We provide a brief background on Bloom filters, and review existing contributions that have applied Bloom filters for network provenance applications. Bloom filters (BF) are probabilistic data structures, which are used to check the membership of an element belonging to a set $\mathcal{S}$ in $O(1)$ (constant amount) time. Specifically, a Bloom filter, denoted by $\mathbf{BF} = [BF[0], BF[1], \ldots, BF[m-1]],$ is an array of $m$ bits such that $BF[i] \in \{0, 1\}$, for $0 \leq i \leq m-1$. Every element of the set $\mathcal{S}$ that wishes to register its membership in the Bloom filter uses an attribute that is unique to itself to choose $k$ random positions in the Bloom filter, for $0 \leq k \leq m - 1$, and then sets the values at those positions to one irrespective of the existing values of $\mathbf{BF}$. One way to achieve this task is to use $k$ hash functions that return $k$ statistically independent numbers in the range $[0, m-1]$. While the elements of $\mathcal{S}$ register their membership this way, a trusted third-party can check the memberships of the elements of $\mathcal{S}$ in a similar way as long as the unique attributes used by the elements are apriori shared with it. In particular, to verify the membership of an element, say $s \in \mathcal{S}$, the third-party computes the $k$ indices of the Bloom filter using the unique attributes of $s$, and then checks whether those $k$ positions are already set in the Bloom filter. If all the $k$ positions are set, then the membership of $s$ is said to be verified. On the other hand, if at least one position is not set, then its membership is not verified. Note that the idea of verifying whether the $k$ locations are set or not, is only a necessary condition but not a sufficient condition. This is because of the possibility of \textit{hash collision} wherein more than one element of the set may use the same set of indices to mark their membership. As a result, even though an element has not embedded its membership in the Bloom filter, its membership may be verified by the third-party. Henceforth, throughout this paper, we refer to such an event as the \textit{false-positive} event. This implies that when using Bloom filters, the choice of $k$ and $m$ must be based on the cardinality of $\mathcal{S}$ such that the occurrence of false-positive events is minimized.

With respect to applications of Bloom filters for provenance recovery, Shebaro et al. in \cite{LDP} have proposed a light-weight secure provenance scheme using Bloom filters that is capable of finding node participation while transmitting a packet. It has been shown that the scheme can detect malicious behaviour such as packet dropping attack. Furthermore, Sultana et al. \cite{LSS} have also addressed a similar problem with the additional capability of detecting packet forgery attacks. Recently, Klonowski et al. have also proposed a light-weight data protocol based on Bloom filters \cite{LWSA} and have analyzed the leakage of information to an adversary that has partial knowledge of the Bloom filter. Although the methods in \cite{LDP}, \cite{LSS}, \cite{LWSA} are effective on static topologies, it is observed that these methods are not applicable in the context of vehicular networks wherein the topology may vary and is unknown at the destination. To specifically handle the framework of vehicular networks, Harshan et al. \cite{BFP} have recently proposed provenance recovery mechanisms, referred to as the edge embedding and the double-edge embedding techniques, and have shown that their techniques offer reduced FPRs in provenance recovery and also reduce latency on the packets. In order to handle mobility of the nodes, \cite{BFP} assumed a complete graph on the underlying topology when designing the Bloom filter parameters and also when implementing the recovery process. As a result, the performance reported in \cite{BFP} is not optimal, especially when the topology is sparse. Pointing at this limitation, we have proposed a two-fold approach of first learning the topology and then using its knowledge in the provenance recovery process. A thorough comparison between our contributions and \cite{BFP} is listed in Table \ref{table:summary}. For literature review on provenance recovery methods that do not use Bloom filters, we refer the reader to \cite{SPS}, \cite{MP}, \cite{PML}, \cite{PDA}, \cite{PAP}, \cite{IPT}, \cite{CPT}, \cite{CTR}, \cite{ntemp}, \cite{paradise}, \cite{step}.

Henceforth, throughout the paper, we refer to a wireless network as a vehicular network if the involved nodes are vehicles. Furthermore, we refer to a vehicular network as a sparse vehicular network if the adjacency matrix of the corresponding graph is sparse. We either use the phrase \emph{vehicular network} or \emph{sparse vehicular network} depending on the context in the rest of the paper.

\section{Network Model on Provenance Recovery}
\label{sec:Network Model}
 
We consider a wireless network comprising $n$ nodes out of which $n-1$ of them are mobile nodes and one of them is the destination. In the context of vehicular networks, the $n-1$ nodes could represent vehicles, whereas the destination could represent a road side unit (RSU). We assume that the set of $n$ nodes can be modelled as an undirected graph $G(N,E)$ where $N=\{1, 2,3, \ldots, n\}$ is the set of nodes such that node $n$ is the destination by default, and $E \subset \{(i, j)~|~ i, j \in N, i \neq j\}$ is the set of edges present in the graph. The existence of an edge, denoted by $(i, j) \in E$ indicates that node $i$ can communicate with node $j$ directly. Since the nodes are mobile and have limited power constraints for transmission, we assume that the underlying topology of the network $G(N,E)$ can vary over time; however, remains fixed for $T$ time units referred to as coherence time of the network. For instance, if the transmission range of a stationary radio device is $1$ Km, then a receiver radio, travelling away from the stationary radio with a velocity of 50 Kmph, can remain in connectivity for about 72 seconds. This in turn implies that the two nodes can communicate roughly 144000 packets under the assumption that each packet is of 0.5 milliseconds (equal to that of a slot as per LTE specifications \cite{standard}). Along similar lines, the maximum of all the pair-wise relative velocities between the vehicles under the RSU will determine the minimum time-duration over which any two nodes will remain in connectivity, and this in turn gives us the coherence time of the network. Henceforth, throughout the paper, we refer the destination and the RSU, interchangeably, and similarly, refer topology and graph, interchangeably.

One of the nodes in the network is a source node that intends to deliver packets to the destination with the help of a subset of other nodes in a multi-hop manner. We assume that a standard routing algorithm such as Ad-hoc On-Demand Distance Vector Routing (AODV) \cite{AODV} is used in the network. During the transmission of a packet from a source node, denoted by node $s$, for $s \in N$, suppose that the packet traverses through the nodes $i_{1},i_{2}, \ldots, i_{h-1}$, for $i_{j} \in N \backslash \{n\}$, before reaching the destination. Since the path chosen by the nodes is ad hoc, the destination does not know the path. As a result, other than routing the packet, the relaying nodes must also assist the destination in learning the path taken by the packet. To achieve this task, each node, when forwarding the packet, also adds its signature on the packet so that the destination, upon receiving the packet, can verify the path traced by the packet using the pre-shared signatures of all the nodes. Henceforth, throughout this paper, \emph{provenance} refers to the information on the path travelled by the packet, i.e., the ordered sequence of nodes that forwarded the packet, \emph{provenance embedding} refers to the process with which every node adds its signature in the packet, and \emph{provenance recovery} refers to the process with which provenance information is recovered at the destination using the signatures of all the nodes.

\subsection{Threat Model}
 
To malign the provenance recovery process, we assume that the network also includes adversaries that execute various security threats such as (i) \textit{Eavesdropping}, wherein an external adversary (an entity outside the network of $n$ nodes) wishes to listen to the transmissions and recover the information on the topology and the provenance, (ii) \textit{Edge-Insertion attack}, wherein an insider (one of the nodes in the network) intends to modify the information in the topology-learning phase such that the destination convincingly learns a non-existing edge in the topology. This way, in the payload phase the insider node can ensure that the packets are flooded through the non-existing edge in the form of a wormhole attack \cite{WHE} without getting detected by the destination. In order to secure the topology learning process and the provenance recovery process against these threats, each node uses a unique secret key (which is already pre-shared with the core network) that is stored in the root of the vehicle akin to the subscriber identity module (SIM) based secret-keys in cellular networks. Furthermore, this root key is used to derive new keys between the node and the core network which in turn would be used \red{to} implement a secure topology learning process and the provenance recovery process. Towards that direction, we first explain a secure neighbour discovery process among the nodes in the following section. 

\subsection{Secure Neighbour Discovery}
 
To design secure provenance embedding algorithms that are suited to the constraints of vehicular networks, we assume that the network includes a number of gateways that are connected to RSUs through a secure backhaul network. We assume that every node (in this case, a vehicle) that enters the network, authenticates itself with a gateway node using a pre-shared root key that is already known to the core network. In this context, the authentication mechanism between a vehicle and the core network could be implemented using a challenge-response strategy. Subsequently, each vehicle derives new keys with the core network by using a standard set of \red{key derivation functions (KDF)} on the root key. Henceforth, we refer to the new set of keys as derived secret keys. Subsequently, every node authenticates with its neighbouring nodes (those nodes that are in its coverage area) using a public-key crypto-system based authentication protocol \cite[Section 2]{BFP}. In particular, when node $i$, for $1 \leq i \leq n - 1$, authenticates with a gateway, it shares its public key \red{ $K_{i_{pub}}$ } along with its unique identity without disclosing the private key \red {$K_{i_{pvt}}$} counterpart of the public key cryptography. Subsequently, this pair (i.e., the public key and the unique identity) of node $i$ is distributed with other nodes by the gateway, either when they authenticate with the gateway, or by broadcasting this information to all the nodes at regular intervals. As a result, when node $i$ wants to authenticate node $j$, it asks node $j$ to share its unique identity by encrypting it using its private key \red{$K_{j_{pvt}}$}. Since node $i$ has the corresponding public key \red{ $K_{j_{pub}}$}, it can decrypt the cipher and then verify the identity of node $j$, thereby completing the authentication procedure. Although vehicular network models (as described in \cite[Section 2]{BFP}) may include multiple RSUs connected to gateway nodes through a core network, we only focus on communication with one RSU in order to explain the intricacies of the provenance embedding algorithms. 


\begin{figure}
\begin{center}
\includegraphics[trim={0 0 0 2cm},clip,scale=0.45]{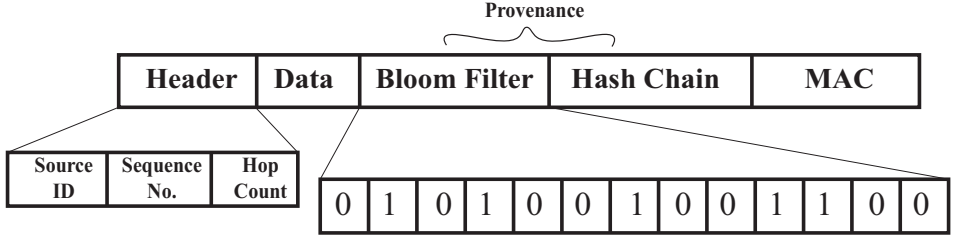}
\end{center}
\caption{A depiction of the packet structure that incorporates a Bloom filter and a hash-chain to assist provenance recovery.}
\label{fig:packet structure}
\end{figure}

\subsection{Objectives}

Based on the description of the network model, it is clear that $G(N, E)$ depends on the mobility pattern and the coverage area of each node. As a result, for a given coherence time of the network, $G(N, E)$ may either be sparse or dense. While the graph's density may dictate how the packets are routed to the destination (because of the ad hoc routing protocol), the destination may not know $G(N, E)$ at the time of provenance recovery. Consequently, the destination may not know the exact set of nodes or the edges to verify, and this in turn may degrade the accuracy of provenance recovery process, especially when $G(N, E)$ is sparse. Identifying this critical relation, we apply Bloom filters for the two-fold purpose of the topology learning phase as well as the provenance embedding/recovery phase. In the former phase, Bloom filters are appropriately used by the nodes to help the destination construct the adjacency matrix of the network, whereas in the latter phase, Bloom filters are used by the nodes to embed the provenance information with special constraints suited to vehicular networks. One of the main assumptions to take this two-fold approach is that the coherence time of the network is sufficiently large in order to accommodate the topology learning phase (for the first fraction of the coherence time) and the payload phase (for the residual fraction). While the use of Bloom filters is fixed throughout the paper, the specific signature used by the nodes to modify the Bloom filters changes as a function of the phase. For instance, the attribute could be a unique identifier associated with a node in the network, an edge in the network, or a double-edge in the network. 

To support the process of provenance embedding and provenance recovery, the packet structure contains three distinct parts, namely: the header, the data and the provenance, as shown in Fig. \ref{fig:packet structure}. The header contains (i) Source ID, (ii) a unique sequence number given to the packet, and (iii) a hop counter which tracks the number of hops covered by the packet. To mitigate security threats, the header part can also be protected by adding a message authentication code using a pre-shared key between the source and the destination. Furthermore, the data portion contains the payload information of the source, and finally, the provenance portion contains a Bloom filter along with an additional space to incorporate hash-chains to improve the reliability of the provenance recovery process. Table \ref{tab:notation} enlists the notations used throughout the paper.

\subsection{Relevance of Bloom filters in Sparse Vehicular Networks}
\begin{table} 
   \centering
   \caption{\label{tab:notation} Notations used in this paper}
   \begin{tabular}{|c|c|}
     \hline
     \textbf{Term} & \textbf{Meaning} \\
     \hline
     $n$ & Number of nodes in the network \\
     \hline
     $\mathbf{BF}$ & Bloom filter array \\
     \hline
     $m$ & Bloom filter size (in bits) \\
     \hline
     $k$ & Number of hash function used \\
     \hline
     $h$ & Number of hops travelled by the packet \\
         \hline
         $seq$ & Sequence number of packet \\
         \hline
         $n_i$ & ID of $i^{th}$ node \\
         \hline
         $\gamma_i$ & Number of neighbours of $i^{th}$ node \\
         \hline
         
         $(i,j)$ & Edge connecting node $n_i$ to $n_j$ \\
         \hline
         $EID_{(i,j)}$ & ID of Directed Edge from $n_i$ to $n_j$\\
         \hline
         $DEID_{(i,j, k)}$ & ID of Directed double-edge $n_i$ $\rightarrow$ $n_j$ $\rightarrow$ $n_{k}$\\
         \hline
         $N$ & Set of all nodes in network\\
         \hline
         $E$ & Set of all edges in network\\
         \hline
         $G(N,E)$ & Graph representing the topology\\
         \hline
         $KDF$ & \red{Key Derivation Function } \\
         \hline
         $K_{i_{root}}$ & \red{Pre-shared secret key between $n_i$ and destination} \\
         \hline
         $K_{i}$ & \red{Derived secret key (KDF($K_{i_{root}}$))}\\ & \red{Used for embedding/recovering the EID/DEID in/from Bloom filter} \\
         \hline
         $K_{i_{pvt}}$, $K_{i_{pub}}$ & \red{Private key  and public key pair used for device-to-device authentication} \\
         \hline
         
     \end{tabular}
 \end{table}

As explained in the objectives, this work makes use of Bloom filters for both the topology learning phase and the payload phase owing to the following benefits: (i) Bloom filters provide a  probabilistic way of embedding the edge identities in the packet without revealing the identities of the edges to the other nodes in (and out of) the network, (ii) The use of Bloom filters forbids the RSU to explicitly share deterministic encoding schemes for embedding the edge identities at regular intervals, and importantly, (iii) Bloom filter helps the RSU to verify the memberships of the edges in the network in O(1) complexity. Despite these advantages, it is clear that the Bloom filter based ideas are applicable as long as the required provenance size in the packet does not exceed the allotted bits in the packet. This observation follows from the fact that the Bloom filter size must increase with the increase in the network size to achieve a given false-positive rate. As a consequence, our framework may not be applicable to sparse vehicular networks when $n$ is large. Therefore, instead of questioning the scalability of our protocols, we believe it is prudent to ask ``What are the various scenarios of vehicular networks when our schemes are applicable?" Towards answering this question, we observe from \cite{rsparse}, \cite{RSUdeploy} that the density pattern of a vehicular network within an RSU depends on two factors: (i) The coverage range of RSUs, depending on whether the network under the RSU is a femtocell or a macrocell, and (ii) Real-time traffic pattern; this depends on the time of the day as well as the city/town wherein the vehicular network is deployed. By broadly classifying the above factors into multiple groups (as shown in Table \ref{scale_table}), namely: short-range and long-range for the coverage area, and very-low traffic, low-traffic and high-traffic for the real-time traffic, we typically have six types of scenarios. Among these, we believe that our schemes are applicable for the following cases: (i) Long-range coverage with very low-traffic. In this context, $n$ is not large and the network can be sparse. (ii) Short-range coverage with low-traffic (also applicable for very low-traffic). However, for long-range coverage with high-traffic, or for short-range coverage with high-traffic, the network may not be sparse, and therefore, the topology learning phase is not required. With the above points, it is clear that our work may not solve the problem for long-range coverage with medium traffic, wherein $n$ can be large and network may be sparse. For such cases,  one may need to think of implementing a)  the proposed schemes with higher density of RSU deployment b)  non Bloom filter based topology learning methods that are space-efficient and yet amenable to low-complexity verification at the RSU.

\begin{table}
\caption{Applicability of our framework in practical scenarios}
\begin{center}
\resizebox{8.5cm}{!}{
\begin{tabular}{|c|c|c|c|}
\hline \textbf{Coverage/Traffic} & \textbf{Very} & \textbf{Low-traffic} & \textbf{High-traffic}\\
& \textbf{low-traffic} & &\\
\hline \textbf{Short range} & applicable & applicable & dense graph \\
& & & (topology learning\\
& & & not required)\\
\hline \textbf{Long range} & applicable & scalability  & dense graph\\
& & problem & (topology learning\\
& & & not required) \\
\hline
\end{tabular}}
\end{center}
\label{scale_table}
\end{table}

\section{Secure Topology-Learning Protocols}  
\label{sec:Toplogy Learning Protocols}

 From the description of the provenance recovery process, it is clear that topology knowledge assists in choosing the right set of parameters for the Bloom filter, especially when the network is sparse. Therefore, the destination has the following tasks: (i) Learn whether the network topology is sparse, (ii) Learn the topology if the hypothesis in (i) is true, and (iii) then use the knowledge of the topology when recovering the provenance. Towards answering the question in (i), the destination can estimate the sparsity of the topology based on the number of neighbours reported by each node. In particular, we propose a framework wherein each node, prior to the topology learning phase, sends a control packet to the destination indicating the number of neighbours connected to it as part of the neighbour discovery process (and not their identities). As a result, with $\gamma_{i}$ denoting the number of neighbours of node $i$, for $1 \leq i \leq n-1$, the destination retrieves the set $\Gamma = \{\gamma_1,\gamma_2, \ldots, \gamma_{n-1}\}$ before the topology recovery phase. From the viewpoint of mitigating integrity threats, $\gamma_{i}$ can be accompanied by a message authentication code using the secret-key between the RSU and node $i$. Subsequently, using this set, the destination then decides whether to learn the topology or not; this is because a dense topology may not significantly reduce the error rates as well as the delay in the provenance recovery process when compared to using the complete graph. Assuming that the topology is sufficiently sparse to learn the topology, we present two topology learning algorithms, namely: (i) Single-Source Multi-Packet (SSMP) embedding, and (ii) Multi-Source Single-Packet (MSSP) embedding. We cover their protocols, derive analytical results on their accuracy, and then discuss their resilience against various threat models. Finally, a thorough comparison of the two algorithms is also provided on various aspects.


\subsection{Single-Source Multi-Packet Embedding}
\label{sec:SSMP}

 In the SSMP technique, each node in the network sends an exclusive packet to the destination in which the Bloom filter portion is embedded with the information of all the neighbours of that node. Upon receiving the packets from all the nodes, the destination verifies the membership of various edges of the complete graph to learn the topology. Since each node uses an exclusive packet to embed its neighbours, the RSU can assign different Bloom filter size for each node since each node may have different number of neighbours. This way, the Bloom filter parameters for the SSMP scheme are $\mathbf{m}=\{m_1,m_2,\ldots, m_{n-1}\}$, which is the list of Bloom filter sizes used by all the nodes, and $\mathbf{k}=\{k_1,k_2,\ldots,k_{n-1}\}$, which is the number of hash functions used by the nodes. It is clear that in order to optimize the accuracy of the SSMP protocol, the Bloom filter parameters $\mathbf{m}$ and $\mathbf{k}$ must be carefully chosen before the learning phase (elaborated in Section \ref{subs:param_sse}). The overall procedure of SSMP embedding is divided into three parts, namely: the embedding process at each node, the packet routing process, and the topology recovery process at the destination.
 
\subsubsection{Embedding and Packet Routing Process}
\label{subsec:SSMP:embedding}
   
Node $i$, for $1 \leq i\leq n-1$ has a unique ID denoted by $n_i$, and a derived secret key (using the pre-shared key with the RSU), denoted by $K_{i}$. By mutual authentication between its neighbour, say node $j$, for $j \neq i$, node $i$ has access to a unique edge ID, referred to as $EID_{(i,j)}$. One way to generate $EID_{(i,j)}$ is to use $n_{i}$ and $n_{j}$, along with key $K_{i}$, to obtain a string $EID_{(i,j)}$ = $Enc_{K_{i}}(n_{i} || n_{j})$, where $Enc$ is an appropriate encryption technique, e.g., AES. Note that this edge ID $EID_{(\cdot, \cdot)}$ is asymmetric in nature i.e., $EID_{(i,j)} \neq EID_{(j,i)}$, where $EID_{(j,i)} = Enc_{K_{j}}(n_{j} || n_{i})$ is the edge ID at node $j$ for the same edge. The set of edge IDs, denoted by $\{EID_{(i, j)}~|~ j \neq i\}$, is also available at the destination since it has access to the IDs and the secret keys of all the nodes. The use of the derived secret key $K_{i}$ ensures that only the RSU is able to learn the topology of the network by verifying the membership of the edges in the Bloom filter. During the embedding process, node $i$, embeds the identity of the edge $(i, j)$ in the Bloom filter as follows: with $seq$ denoting the sequence number of the packet, node $i$ generates a string $EID_{(i,j)}||seq$, where $||$ denotes the concatenation operation. Subsequently, using the string and a set of pseudo-random nonce values, node $i$ uses a hash function to generate $k_{i}$ indices on the Bloom filter of length $m_{i}$, and then sets those indices to one. This process is repeated for embedding every edge with its neighbour. A detailed description of the embedding process is also depicted in Fig. \ref{fig:Single_source_embedding2}.
  
After the neighbour embedding process, each node routes its packet to the destination in an ad hoc manner with minimal overhead in terms of the number of hops. Note that only one node embeds the packet while the others facilitate to forward the packet to the destination. In this context, AODV \cite{AODV} can be used. 

\subsubsection{Topology Recovery Process}
  
In the topology recovery process, the destination constructs the topology of the network using an $n \times n$ adjacency matrix, denoted by $ADJ$, wherein $ADJ[i][j] = 1$ and $ADJ[i][j] = 0$, depending on whether the edge $(i, j)$ exists or not, respectively. When the packet from node $i$ is received, for $1 \leq i \leq n - 1$, the destination extracts the Bloom filter portion from it and then uses the knowledge of the identities of its $n-1$ edges to verify their participation in the Bloom filter. If an edge $(i, j)$ is verified, i.e., if the locations chosen by the edge $(i, j)$ are set to one in the Bloom filter, then $ADJ[i][j]$ is updated to $1$. Since a non-existing edge can also be verified by the destination (due to the probabilistic nature of Bloom filters), we incorporate a reinforcement check in the end wherein the edge $(i, j)$ exists in the topology only when both $ADJ[i][j]$ and $ADJ[j][i]$ are set to one. Even if one of them is zero, such an edge will be discarded from the topology.

\begin{figure}
\begin{center}
 \includegraphics[trim={0 5cm 0 0},clip,scale=0.4]{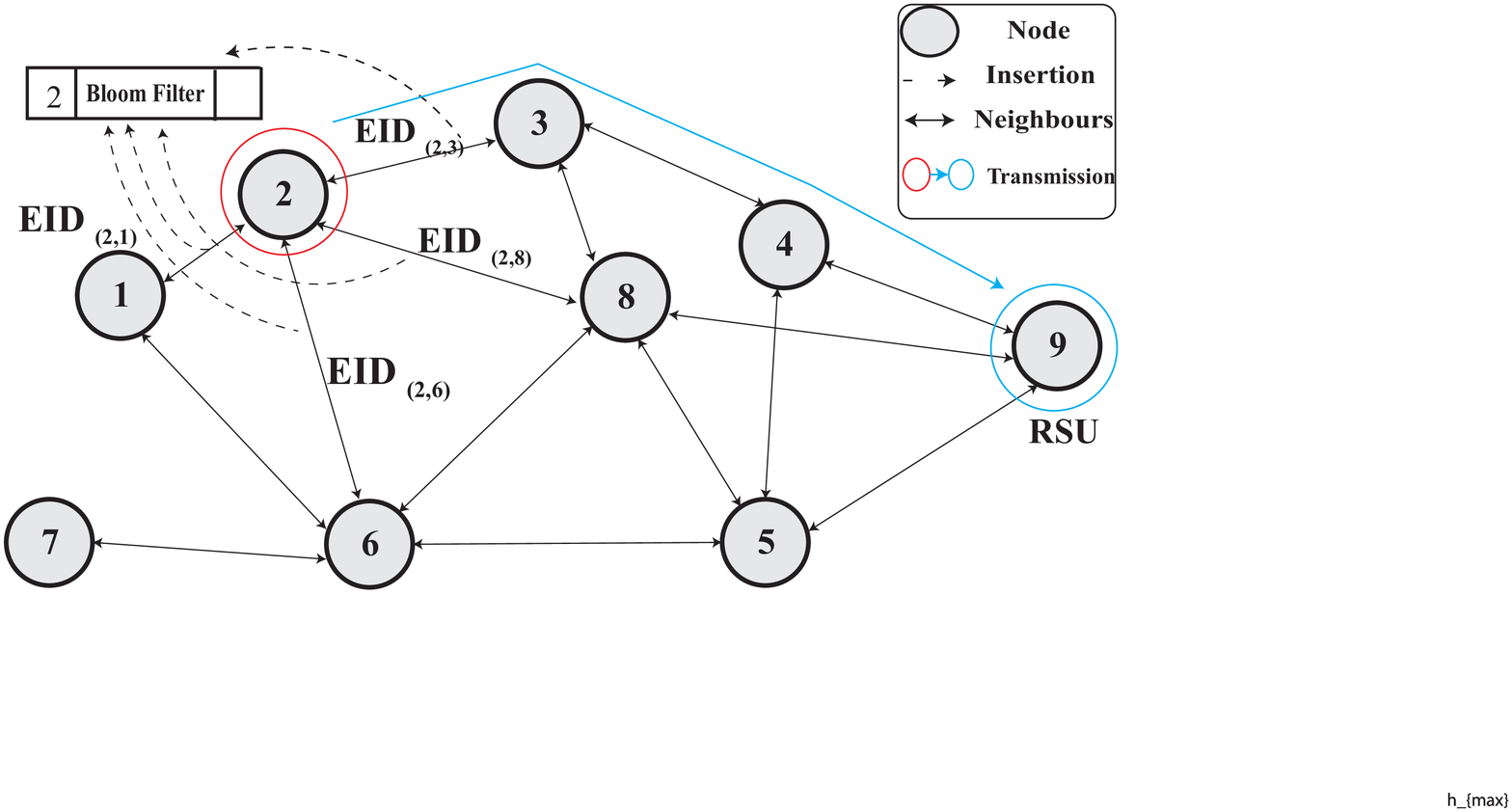}
\end{center}
\caption{SSMP technique where node 2, upon its turn, creates a packet with source $ID = 2$, embeds its neighbours in the form of EID in the Bloom filter, and then transmits it to the RSU.}
       \label{fig:Single_source_embedding2}
   \end{figure}

  \comment{ 
  \begin{algorithm}
\label{al:Topology_learning_using_Single_Source_Embedding}
\SetAlgoLined
\textbf{Input}:{
\\

\textbf{1)} IDs of all nodes i.e. $n_1,n_2,n_3...n_{n-1}$\;
\textbf{2)} Packet from all nodes except from RSU i.e.$P_1,P_2,P_3,P_4...P_{n-1}$ each $n_i$ having Bloom filter $BF_i$ of size $m_i$, and with sequence numbers $seq_1,seq_2,seq_3,...,seq_{n-1}$\;
\textbf{3)} $H_m(.)$ is the hash Function that generate a value between $1$ to $m$;\\
}
\vspace*{1\baselineskip}
\textbf{Output}:{\\
Set of all directed edges present in the topology.\\
}
\vspace*{1\baselineskip}
\textbf{Initialize:}{\\
\textbf{1)} ADJ=$N\times N $ adjacency matrix with all values initialized to 0\;
\textbf{2)} Edges=$\phi$
}
  \vspace*{1\baselineskip}\\
1. Find the connectivity of a node $n_i$ to all other nodes $n_j$ in the topology. \\
\For{i from 1 to n-1}{
\For{j from 1 to n-1}{
\If{$i\neq j$}{
$EID=generateEID(E(i,j),seq_i)$\;
$Hash=\phi$\;
\For{$r$ from $1$ to $k$}{
    $Hash=Hash \cup H_{m_i}(EID||r)$\;
}
\If{Every index of the Hash in $BF_i$ is set to 1}{
$ADJ[i][j]=1$\;
}

}
}
}

2. Reinforce the learning by mutually confirming the neighbours \\
\For{i from 1 to n-1}{
\For{j from i+1 to n-1}{
\If{$ADJ[i][j]\neq ADJ[j][i]$}{
$ADJ[i][j]=0$\;
$ADJ[j][i]=0$\;
}
}}
3. Extract the valid edges from the adjacency matrix
\For{i from 1 to n-1}{
\For{j from 1 to n-1}{
\If{ADJ[i][j]=1}{
$Edges=Edges \cup (i,j)$
}
}}
4. Add all the edges that are terminating in the \st{base station} road side unit.\\
\Return $Edges$
 \caption{Topology learning using SSMP}
\end{algorithm}
}

\subsubsection{False-Positive Analysis}
\label{subs:param_sse}
From the topology recovery process, it is clear that all the edges reported by the nodes will be registered in the adjacency matrix. However, owing to the probabilistic nature of Bloom filters, there exists a non-zero probability with which a non-existing edge is also recovered. Henceforth, we refer to such an event as the false-positive event, denoted by $fp$, and formally define its rate as FPR, given by $P(fp) =\frac{N_{error}}{N_{total}},$ where $N_{total}$ denotes the total number of times the SSMP protocol is executed, and $N_{error}$ is the number of times the destination recovers non-existing edges in the topology. This implies that to learn a given topology, denoted by $G(N,E)$, the Bloom filter parameters $\mathbf{m}=\{m_1,m_2,m_3, \ldots, m_{n-1}\}$ and $\mathbf{k}=\{k_1,k_2,k_3, \ldots, k_{n-1}\}$ must be chosen such that the FPR is minimized. However, we note that running a large number of experiments and simulations to compute the Bloom filter parameters is not computationally feasible. To circumvent this problem, we derive closed-form expression for $P(fp)$, which in turn would be useful to choose the Bloom filter parameters.
    
\begin{theorem}
\label{thm:SSMP_exact}
The closed-form expression of $P(fp)$ for a network topology $G(N,E)$, $\mathbf{m}$ and $\mathbf{k}$ is given as:
$P(fp)= P\biggl(\bigcup_{(a,b) \in \overline{E}} E_{a \leftrightarrow b}\biggr)$, where $E_{a \leftrightarrow b}$ is defined in \eqref{eqn:2} and $\overline{E}$ is the edges in the complementary graph of the topology excluding the edges that are incident to RSU.
\end{theorem}
\begin{proof}
We refer the reader to Proof \ref{prf:SSMP_exact} in the appendix section.
\end{proof}

Using the above theorem, we can analytically obtain the FPRs as a function of the Bloom filter size and the number of hash functions. However, an issue with this approach is that it requires the knowledge $G(N,E)$, which in turn is not aligned with the objective of the topology learning phase. As a result, we show that $\Gamma = \{\gamma_{i}, 1 \leq i \leq n-1\}$ can be used to derive an upper bound on the FPRs. Consequently, we can use the upper bound as the objective function to obtain $\mathbf{k}$ and $\mathbf{m}$ such that $\sum_{i = 1}^{n-1} m_{i} = m_{sum}$.

\begin{theorem}
\label{thm:SSMP_bound}
When the topology of the network is not known, an upper bound $P'(fp)$ of $P(fp)$ can be obtained by using $\Gamma$.
\end{theorem}
\begin{proof}
We refer the readers to Proof \ref{prf:SSMP_bound} in the appendix section.
\end{proof}
The above theorem implies that despite not knowing the network topology, the destination has an upper bound on the FPRs that is only a function of $\Gamma$, $\mathbf{m}$, and $\mathbf{k}$. Thus, it can solve Problem \ref{problem:2} to obtain the Bloom filter parameters.


\begin{problem_stmt}
\label{problem:2}
Given $\Gamma$ and $m_{sum}$, solve
$$
\tilde{\mathbf{k}}^{*},\tilde{\mathbf{m}}^{*}=\argminA_{\{k_1,k_2, \ldots, k_{n-1}\}\{m_1,m_2, \ldots, m_{n-1}\}} P'(fp),
$$
 subject to $1\leq k_i\leq m_i$ and $\sum_{i = 1}^{n-1} m_i=m_{sum}$, where $P'(fp)$ is given in \eqref{eq:upper_bound}.
\end{problem_stmt}

\subsubsection{Simulation Results on Parameter Optimization}

In this section, we present experimental results to validate the results of Theorem \ref{thm:SSMP_bound}. To carry out the experiments, we consider a sparse vehicular network comprising $n = 8$ nodes with $14$ edges, and this network  will  be  used  for  generating the  simulation results throughout this section on topology learning protocols. For the SSMP protocol, although the FPRs and their upper bound are functions of $2(n-1)$ variables, namely: $n-1$ values of $\mathbf{m}$ and $n-1$ values of $\mathbf{k}$, we use equal sized Bloom filters for every node, i.e., $m_{i} = \frac{m_{sum}}{n-1}$, $\forall i$, for the purpose of demonstration. As a result, the number of hash functions at each node is also equal, i.e., $k_{i}=k, \forall i$. Fixing a given value of $m$, in Fig. \ref{fig:upper bound}, we plot the FPRs and their upper bound by varying the value of $k$. The plots confirm the upper bound behaviour, and more importantly, we observe that the value of $k$ for which the exact FPR is minimized is approximately same as the value of $k$ that minimizes the upper bound. This implies that, the destination can solve Problem \ref{problem:2} to obtain the Bloom filter parameters for the topology learning case especially when the Bloom filter sizes at all the nodes are equal. However, we point out that to showcase similar results for variable size Bloom filters (when the sizes of Bloom filters are different at the nodes) is challenging as the search space for $\mathbf{k}$ and $\mathbf{m}$ is too large to verify the solutions to Problem \ref{problem:2}. 

In the rest of this section, we present simulation results to show that allocating variable size Bloom filters is better for the SSMP protocol when the number of neighbours reported by the nodes are different, i.e., $\gamma_{i} \neq \gamma_{j}$, for some $i \neq j$. To present these results, we consider a sparse vehicular network consisting of $n = 8$ nodes such that the neighbour distribution reported by the nodes is given by $\Gamma = [5, 3, 4, 1, 4, 2, 4 ]$. For this network, first, we fix a value of $m_{sum} = 280$, and then obtain the FPRs for the case of equal sized Bloom filters, i.e., when $m_{i} = 40$ for each $1 \leq i \leq 7$. The plot of the FPR for the above case is presented in Fig. \ref{fig:equal_var}, which shows that $k = 6$ minimizes the FPRs. Now, we explore whether variable size Bloom filters would yield lower FPRs than the one offered by the equal size Bloom filters. To verify this possibility in a computationally-friendly manner, we consider all possible ways of distributing a total of $m_{sum} = 280$ (35 bytes) across the $7$ nodes, as this would give us the overall search space for $\mathbf{m}$. However, given the complexity constraints, instead of distributing 35 bytes among the $7$ nodes at the granularity of bits, we incorporate the idea of distributing $36$ bytes among $7$ distinct nodes with the granularity of $2$ bytes such that each node gets at least 2 bytes and size of Bloom filter allotted to each node is changed with the granularity of 2 bytes. With this reduced search space for $\mathbf{m}$, we then proceed to solve Problem \ref{problem:2} to compute $\mathbf{k}$ that minimizes the upper bound on the false-positives for a given $\mathbf{m}$. After solving the conditional minimization problem for a given $\mathbf{m}$, we repeat the process over the reduced search space of $\mathbf{m}$, and then pick that pair $(\hat{\mathbf{m}}, \hat{\mathbf{k}})$ that minimizes the upper bound for this exercise. In Fig. \ref{fig:equal_var}, we also plot the exact FPRs of the SSMP protocol using the pair $(\hat{\mathbf{m}}, \hat{\mathbf{k}})$. The plot shows that the FPRs offered by variable size Bloom filters is significantly lower than that of equal size Bloom filters for the same $m_{sum}$. We highlight that the actual benefits of variable size Bloom filters can only be studied by solving Problem $\ref{problem:2}$ over the entire search space in an exhaustive manner. In summary, we have shown that the right choice of the Bloom filter parameters can be made using the upper bound provided in Theorem \ref{thm:SSMP_bound}, which only requires the knowledge of the neighbour distribution of the nodes and not the exact topology of the network. 

\begin{figure}
    \centering
     \includegraphics[trim={0 0 0 5cm},clip,scale=0.2]{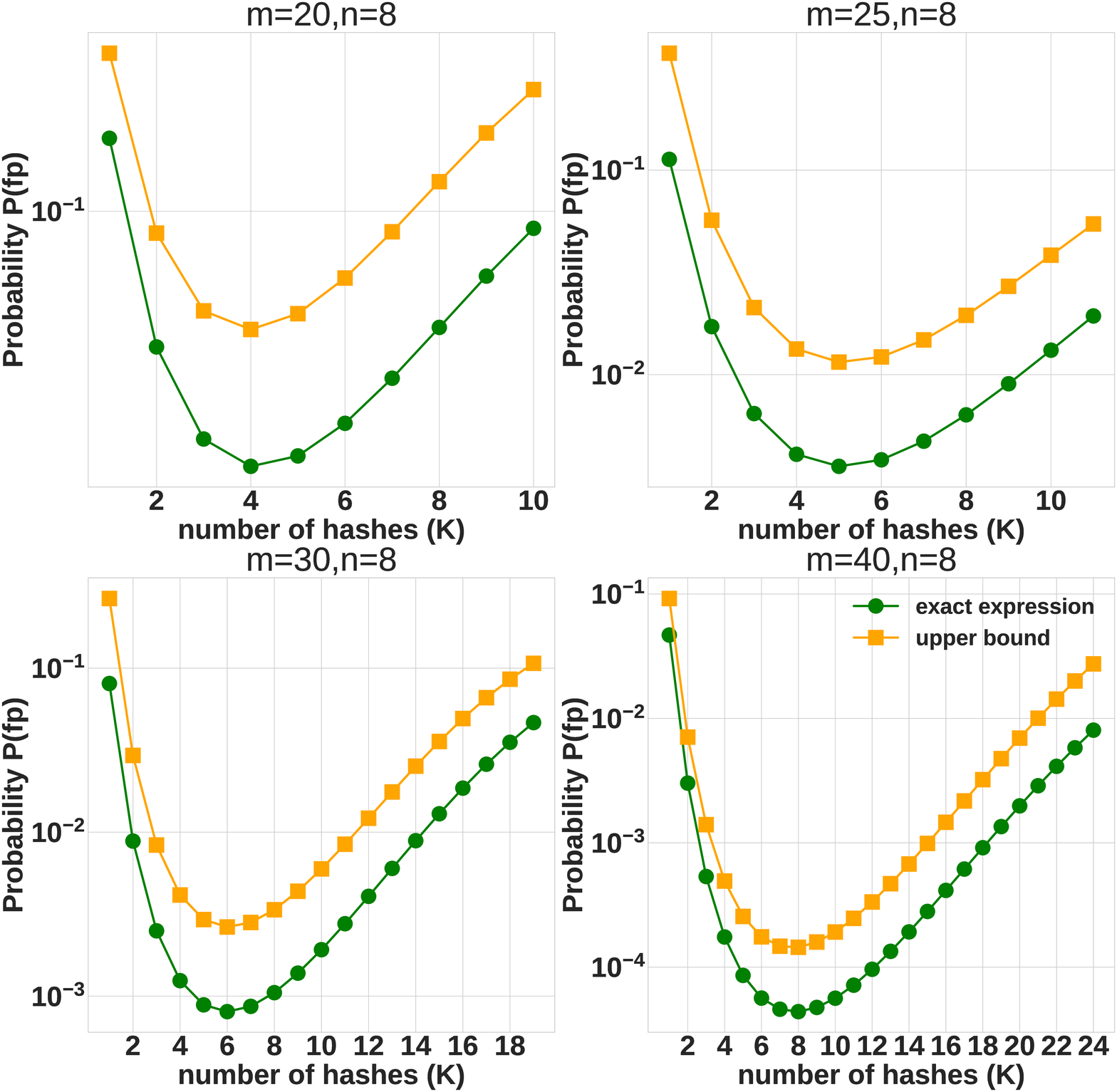}
\caption{Plots for the proposed upper bounds on FPRs of the SSMP technique for a network with $n=8$. The plots confirm that the bounds provide the same minima in terms of $k$ without the topology knowledge.}
    \label{fig:upper bound}
\end{figure}

\begin{figure}
    \centering
     \includegraphics[trim={0 0 0 4cm},clip,scale=0.3]{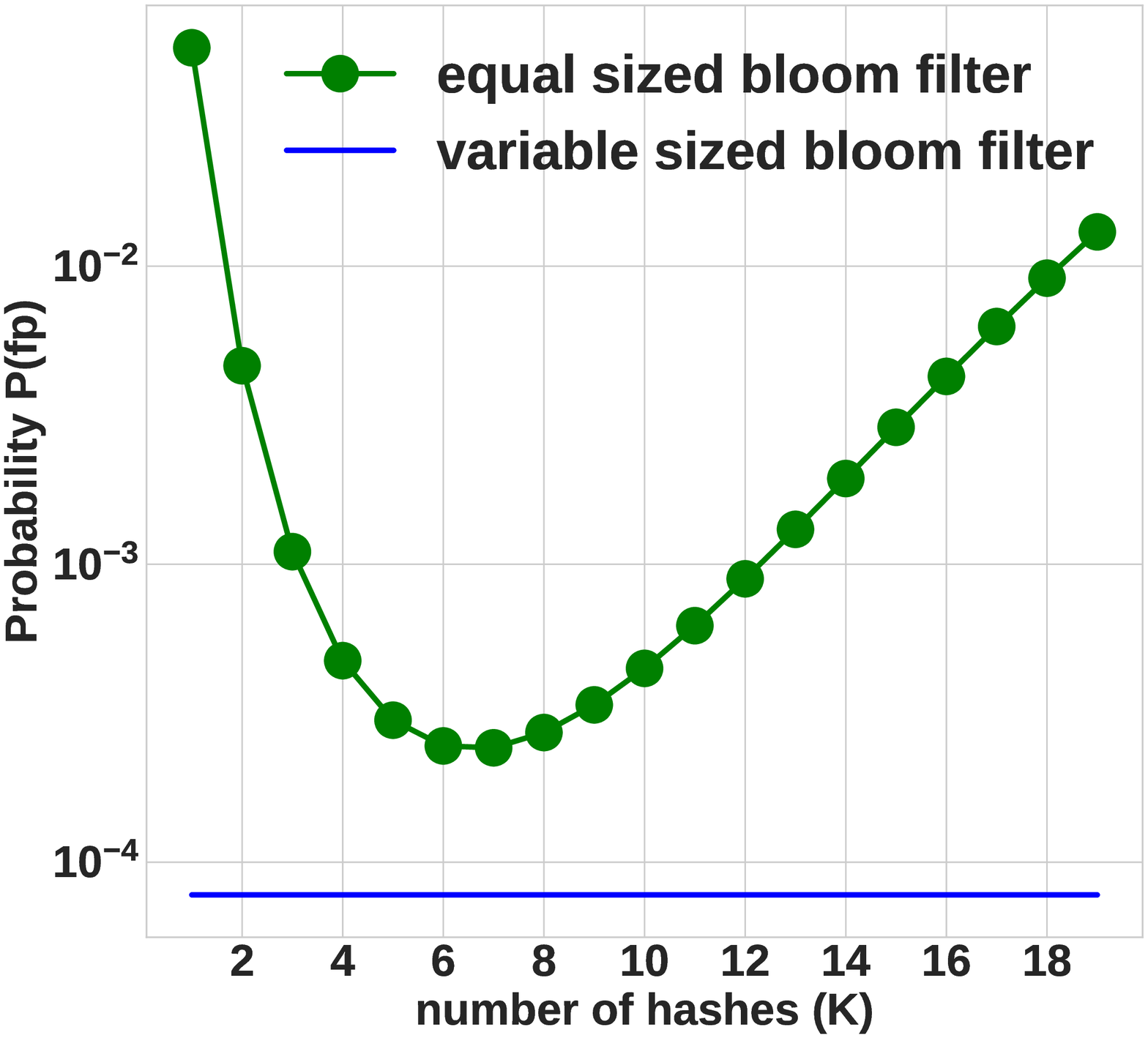}
\caption{FPRs of the SSMP protocol with equal sized Bloom filter and variable sized Bloom filters for the nodes. For a network setting with $\Gamma=[5, 3, 4, 1, 4, 2, 4 ]$ and $m_{sum}=280$, we introduced a fixed size Bloom filter for all nodes i.e., $m_i=40$, and compared that with variable sized Bloom filter given by $\mathbf{m} = [48,48, 48, 16, 48, 32, 48]$ with number of hashes as $\mathbf{k} = [8,10, 8, 7, 8, 9, 8]$. The plots show that variable sized Bloom filters outperform the equal sized case. Note that horizontal line is used only to indicate the FPR of the variable sized Bloom filter.}
    \label{fig:equal_var}
\end{figure}

\subsection{Multi-Source Single-Packet Embedding}

In the SSMP protocol, although a number of nodes routed the packets to the destination, they do not utilize this opportunity to embed the information on their neighbours. As a result, the number of packet transmissions across the nodes may be an overhead. To circumvent this problem, the MSSP protocol focuses on embedding the neighbours of all the nodes in one single packet. This gives a clear advantage over the SSMP embedding method as it requires fewer packet transmissions, thereby increasing the energy efficiency of the wireless nodes. While the MSSP protocol seems to reduce the number of packet transmissions, its accuracy of the topology learning process needs to be studied since a given space of Bloom filter portion is accessed by all the nodes in the network instead of just one node. Towards studying these features, in the following sections, we introduce the MSSP protocol by describing the procedure for embedding the neighbours, its routing protocol and the topology recovery process. We specifically focus on deriving analytical expressions for the FPR so that we can compare that with the SSMP protocol.
 
\subsubsection{Embedding and Packet Routing Process}
   
The embedding process of the MSSP protocol (as depicted in Fig. \ref{fig:MSSP2}) is similar to that of the SSMP protocol. However, given that all the nodes embed the information of their neighbours on one packet, the final stage of setting the Bloom filter values undergoes some changes. In particular, node $i$ upon generating the string $EID_{(i,j)}||seq$ to embed the information on its edge with node $j$, uses $k$ Hash functions to generate the indices on the Bloom filter of length $m$, and then sets those indices to one. In contrast with the SSMP protocol, the Bloom filter size used by all the nodes is the same, and moreover, all the nodes use $k$ hash functions to generate the indices (instead of $k_{i}$ indices).

For the MSSP protocol, the packet routing is done using a dedicated algorithm that minimizes the number of transmissions across the nodes such that the packet traverses all the nodes before reaching the destination. For instance, a tree-based protocol \cite{TBR} will ensure that the packet reaches the destination with minimal overheads without forming any loops within the network. Once a node in the topology receives the packet from any of the neighbouring nodes, it will have to check whether it has already embedded its neighbours in the packet; this is because the packet may reach a node more than once en-route to the destination.
    
\subsubsection{Topology Recovery Process}
    
The topology recovery process of the MSSP protocol is similar to that of the SSMP protocol. However, given that all the nodes embed the edges with their identities on one packet, the destination verifies the presence of all the edges of the complete graph, and then generates the adjacency matrix. Similar to the SSMP protocol, an edge $(i, j)$ is said to be present in the topology only if both node $i$ and node $j$ have embedded the edge $(i, j)$ in the Bloom filter. 
    
   
     \begin{figure}
     \centering
       \includegraphics[trim={3.55cm 0 0 0},clip,scale=0.4]{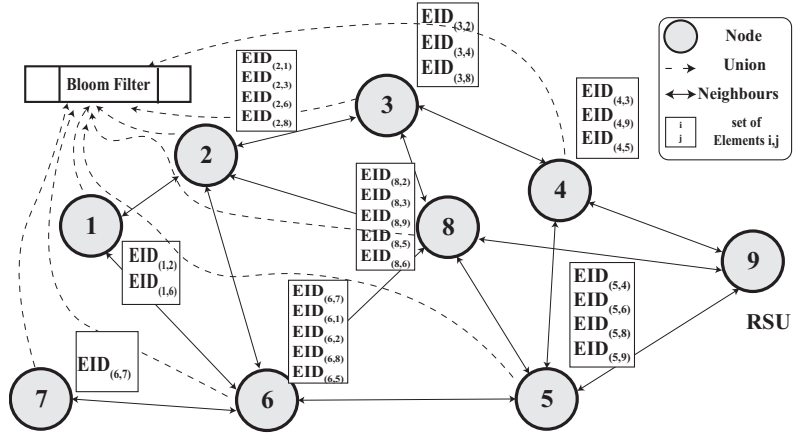}
        \caption{A depiction of MSSP technique where each node in the topology will insert the EID of its neighbour to a single Bloom filter.}
        \label{fig:MSSP2}
    \end{figure}
    
\subsubsection{Optimization of Bloom Filter Parameters}
    
    \label{sub_sec:param ope}
    Similar to the SSMP algorithm, non-existing edges in the topology may also be recovered as false-positive events. Thus, for the MSSP protocol, the Bloom filter parameters $m$ and $k$ must be chosen such that the FPRs are minimized. Towards that direction, the following theorem provides a closed-form expression on the FPR of the MSSP protocol.
     \begin{theorem}
     \label{thm:MSSP}
     The closed-form expression for the FPR of the MSSP protocol can be computed as a function of $k$, $m$, and the total number of edges embedded by all the nodes. 
     \end{theorem}
     \begin{proof}
     We refer the reader to Proof \ref{prf:MSSP} in the appendix section.
     \end{proof}
From the expression of the FPR, it is clear that $P(fp)$ is a function of $k$, $m$, $2|E|-\gamma_{RSU}$ and $|\overline{E}|$. Since $\gamma_{RSU}$ is already known to the destination, the destination must somehow learn $|E|$ and $|\overline{E}|$. The following proposition (which can be proved in a straightforward manner) shows that the destination can compute $|E|$ and $|\overline{E}|$ by using the knowledge of $\{\gamma_1,\gamma_2,\ldots,\gamma_{n-1}\}$, which are shared by each node in the beginning of the topology learning phase. 
\begin{proposition}
Using $\Gamma = \{\gamma_1,\gamma_2,\ldots,\gamma_{n-1}\}$ and $\gamma_{RSU}$, we have $|\overline{E}|={{n}\choose{2}}-{\frac{\gamma_{RSU} + \sum_{i=1}^{n-1}\gamma_i}{2}}-n+\gamma_{RSU}+1,
$ and $2|E|= \gamma_{RSU} + \sum_{i=1}^{n-1}\gamma_i.$
\end{proposition}

Using the above proposition, the destination can choose the Bloom filter parameters by solving Problem \ref{problem_stmt_MSSP} without knowing the topology. 
{\begin{problem_stmt}
\label{problem_stmt_MSSP}
For a given $\{\gamma_1,\gamma_2, \ldots,\gamma_{n-1}\}$, $\gamma_{RSU}$, and a Bloom filter size $m$ solve 
\\
$$
k^*=\argminA_{1 \leq k \leq m} P(fp).
$$

\end{problem_stmt}}
We highlight that the objective function of the optimization problem in Problem \ref{problem_stmt_MSSP} continues to be the exact expression of the FPR, and this is one of the benefits of the MSSP protocol over the SSMP protocol. 

\begin{table*}
     \centering
     \caption{A brief summary of comparison between the SSMP and the MSSP protocols.}
     \resizebox{16cm}{!}{
     \begin{tabular}{|c|c|c|}
     \hline
        \textbf{Metric for evaluation} & \textbf{SSMP} & \textbf{MSSP} \\
         \hline
          False-positive rates & Better than MSSP, can be used & Not as good as SSMP \\
          &when ultra-reliability is a key requirement&\\
          \hline
          Transmission overhead & High & Low \\
          \hline
          Parameters optimization & Can find upper bound & Can find closed-form expression\\
          & when the topology is not known & when the topology is not known\\
          \hline
          Network congestion & High & Low \\
          \hline
          Coordination for implementation & Low: any node can start the process & High: nodes must decide the embedding order\\
          \hline
     \end{tabular}}
     \label{tab:comp_table_learning}
 \end{table*}
 
 \subsection{Comparison of Topology Learning Protocols}
 \label{sec:comparison}
 
For a network with $n = 8$ and $14$ edges, we compare the FPRs of the two algorithms in Fig. \ref{fig:my_label} when the total Bloom filter size is the same, i.e., the Bloom filter size for each node in the SSMP protocol is $m_i=\frac{m_{sum}}{n-1}$, whereas the Bloom filter size of the MSSP protocol is $m=m_{sum}$. The plots in Fig. \ref{fig:my_label} confirm that the SSMP protocol outperforms the MSSP protocol even though the Bloom filter size distribution across the nodes is equal (which may be sub-optimal). Thus, if ultra-reliability in false-positives is required, then the SSMP protocol must be preferred. 
 \begin{figure}
     \centering
     \includegraphics[trim={0 0 0 5cm},clip,scale=0.16]{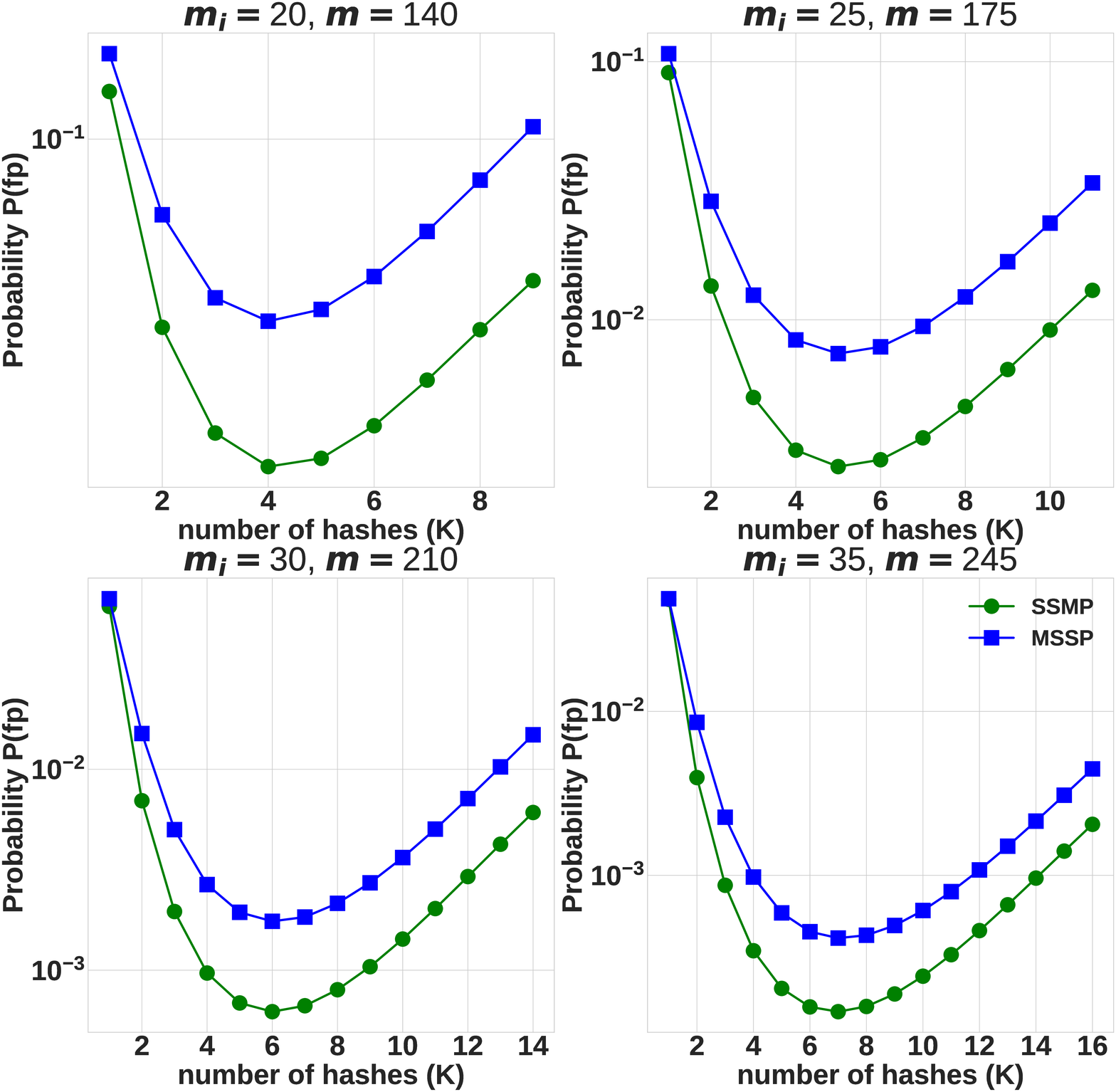}
     \vspace{-1cm}
     \caption{Comparison between SSMP and MSSP techniques in terms of FPRs, with equal sized Bloom filter for SSMP. The Bloom filter size for MSSP protocol is equal to sum of  Bloom filter sizes of all the nodes in the SSMP technique i.e., $m=m_{sum}=\sum_{i = 1}^{n-1} m_i$.} 
     \label{fig:my_label}
 \end{figure}

Finally, while the SSMP protocol outperforms the MSSP protocol in terms of FPRs, the former algorithm has poor load distribution on the nodes during the routing protocol; this is because the nodes geographically closer to the destination have to repeatedly forward the packets of other nodes, thereby resulting in skewed distribution on the number of packet transmissions across the nodes. A summary of comparison between the topology learning protocols is also provided in Table \ref{tab:comp_table_learning}, capturing their FPRs, transmission overheads, the ease of estimating Bloom filter parameters, load distribution during the routing process, and implementation complexity.

\subsection{Security Analysis}
\label{sec:threat}

Motivated by providing a secure provenance method, the topology learning phase handles confidentiality and integrity threats through the use of (i) pre-shared secret key between the nodes and the core network for authentication, (ii) public-key cryptography based authentication between nodes during the neighbour discovery protocol, (iii) message authentication codes for sharing the information on the number of neighbours, (iv) encryption algorithms using shared keys between the nodes to generate the edge identities and the double-edge identities, and finally, (v) Bloom filters to embed the information on the edges. Since the above measures take care of threats from external nodes (nodes outside the network), in this section, we focus on insider attacks, i.e., attacks executed by one or more nodes that are part of the network. Within this class, we address eavesdropping on the identity of the edges connected to the other nodes, and the edge-insertion attack: an integrity threat using which an internal node modifies the protocol so that the RSU learns non-existing edges in the network. This way, the attack is relevant to this work as it attempts to reduce of sparsity of the network, thereby maligning the original objective of capitalizing on sparse topology for provenance recovery. Towards mitigating eavesdropping from insiders, the idea of generating edge identities using an encryption algorithm and its subsequent use as input to the Bloom filter solves the problem. Towards mitigating the edge-insertion attack, the idea of using message authentication code for sharing the neighbour information forbids nodes to modify others' information. Also, the idea of using reinforcement check of registering an edge in the topology only when both its vertices embed that edge in the Bloom filter, prevents an insider from adding a non-existing edge in the Bloom filter. While these methods prevent easy execution of edge-insertion attack, we also study their sophisticated variants in the following section.

\subsubsection{Edge-Insertion Attack}
\label{sec:Integrity Threat}


In this section, we study the feasibility of edge-insertion attacks on the proposed algorithms in the topology learning phase. In particular, let us consider an insider attack, wherein one of the nodes, say node $d$, for some $d \in N$, wishes to add a non-existing edge, say with node $x$, in the topology. To successfully achieve this attack without generating a false-positive event at the destination, the adversary must complete the neighbour discovery protocol with node $x$ (that is not one of its neighbours) possibly through a wormhole attack \cite{WHE} with the help of another node as its proxy. This way, both node $x$ and node $d$ would report this edge in the vector $\Gamma$ as well as in the Bloom filter. However, if the wormhole attack cannot be implemented during neighbour discovery, then the number of neighbours reported by node $x$ would be different from that recovered from the Bloom filter. As a result, a false-positive event would be generated. Thus, we claim that an insider node cannot add a new edge in the topology without implementing a wormhole attack in the neighbour discovery phase.
    
In the rest of this section, we first explain the detailed process of implementing the above discussed edge-insertion attack, and also point out potential solutions that could be used to mitigate this attack. We note that an edge-insertion attack requires the adversary to compromise two phases, namely: the neighbour discovery phase and the topology learning phase. Before the beginning of the neighbour discovery phase, let us assume that one of the nodes in the network, say node $d$, compromises the credentials of node $x$. For instance, we could assume that the adversary gets hold of both the private key (used for public-key cryptography based authentication) and the derived key (used for embedding the identity of the edges in the SSMP and MSSP algorithms). Subsequently, node $d$ will try to advertise itself as node $x$ to another node, say node $y$, in the neighbour discovery phase through the help of an external node by forming a wormhole. Once node $d$ authenticates as node $x$ with node $y$, node $y$ will embed the edge $(x, y)$ in the Bloom filter in the topology learning phase. However, to force the RSU to register the edge $(x, y)$ in the topology, node $d$ also needs to modify the Bloom filter sent by node $x$ by adding the information on the edge $(x, y)$. In order to achieve this last task, the adversary will have to use the derived key that has already been compromised. Thus, by compromising the credentials of node $x$, the adversary can execute an edge-insertion attack through a wormhole. 

One easy way of preventing this edge-insertion attack is to ensure that both the private key and the derived key at a node are securely stored in the root of the device. However, given that these keys are often retrieved and stored in the memory during the neighbour discovery phase (the private key $K_{i_{pvt}}$ in this case) or the topology learning phase (the derived key $K_i$ in this case), they are vulnerable for compromise through side-channel attacks by the attacker. Therefore, we need to consider the possibility of compromising the following different combinations of keys, namely: (i) only the derived key $K_i$ is compromised, (ii) only the private key $K_{i_{pvt}}$ is compromised, and finally, (iii) both the private key and the derived keys are compromised. Under case (i), i.e., when only the derived key is compromised, we note that the edge-insertion attack is not possible since the neighbour discovery phase is already secure due to secure private key. Under case (ii), i.e., when only the private key is compromised, while the neighbour discovery phase can be compromised, the attacker cannot add the information on the edge $(x, y)$ in the Bloom filter as it does not know the derived key. Although the attacker does not know the derived key, it can randomly choose some locations on the Bloom filter and set them hoping that they would coincide with the locations that would be chosen by node $x$ when embedding the edge $(x, y)$ in the Bloom filter. Formally, when using a random attack, node $d$ is said to successfully execute an impersonation attack during the topology learning phase if the Bloom filter locations it chooses match the locations that would be chosen by node $x$ when embedding the edge $(x, y)$ in the Bloom filter. 

\begin{proposition}
\label{prop:security_imp_attack}
When using the MSSP algorithm, the success-rate of perfect impersonation attack is non-zero.
\end{proposition}
\begin{proof}
We refer the reader to Proof \ref{prf:security_imp_attack} in the appendix section.
\end{proof}

Although the above result shows that the probability of successful impersonation attack is non-zero, it can be verified that the probability of these events is negligible as long as the Bloom filter size is sufficiently large. Therefore, despite using the sophisticated wormhole attack by compromising the private key, our framework ensures that the success-rate under case (ii) is negligible. Finally, when addressing case (iii), we need to secure the neighbour discovery phase wherein the adversary must not be able to impersonate by compromising the private key of the victim. Towards that direction, we could replace the public-key cryptography based authentication in the neighbour discovery phase with device fingerprinting based authentication mechanisms \cite{device_fingerprint}. With such ideas, a secret key is derived using the unique perturbations arising out of the physical properties of each device, and subsequently these unique features can be verified by the other devices that gather this information when they authenticate with the gateway nodes upon entering the network. This implies that even if the private key stored on the device memory is compromised by the attacker, its unique physical properties cannot be cloned by the adversary.

In general, when multiple nodes within the network collude, they may execute generalized versions of wormhole attack during the neighbour discovery protocol, thereby adding several non-existing edges to the topology. As a result, we recommend that state-of-the-art mitigation techniques against wormhole attacks \cite{WHE} must be employed. 

 

\section {Payload Phase with Topology Knowledge}
\label{sec:payload}

Once the topology of the network is recovered in the first part of the coherence time, the next part is used for the payload phase, wherein one of the nodes communicates its data to the destination in a multi-hop manner. As a result, the payload part of the packet carries the data of a source, whereas its provenance portion carries information on the path travelled by the packet. Using \cite{BFP}, we apply a variant of deterministic edge embedding (DE) and deterministic double-edge embedding (DDE) algorithms to assist provenance recovery at the destination. Although \cite{BFP} introduced the DE and DDE methods to resolve paths due to no knowledge of topology at the destination, we observe that these methods continue to help in resolving paths with topology knowledge especially when the topology contains cycles. However, unlike \cite{BFP}, since topology knowledge is available during provenance recovery, we expect to eliminate non-existing edges (or double-edges) thereby improving the accuracy of the provenance recovery process. 

Since fixed-size Bloom filters are used by each embedding node, we expect the rate of false-positive events to increase as the network size increases. Here, a false-positive event refers to a scenario when the destination encounters more than one path of a given hop-length when recovering the provenance. Therefore, only using Bloom filters for the provenance recovery process increases the communication-overhead of the protocol, i.e., the value of $m$, to achieve a given FPRs. To circumvent this problem, unlike \cite{BFP}, we supplement the use of Bloom filters by using an in-packet \textit{hash-chain} that assists in path verification on a packet-to-packet basis. With the use of hash functions such as SHA-256, we expect that hash-chain guarantees negligible probability of collision, thereby helping the destination to resolve the candidate paths obtained from the Bloom filter with an overwhelming probability. Furthermore, we observe that to achieve a given FPR, the number of Bloom filter bits can be reduced provided the destination has the complexity to verify a large number of candidate paths. On the other hand, the number of Bloom filter bits can be increased provided the destination is computationally bounded to perform hash-chain verification on a large number of paths. Thus, by using $\beta$ as the affordable number of hash-chain verification by the destination, we observe a trade-off between $m$ and $\beta$ to achieve a given FPR. With the above mentioned trade-off, in the rest of the section, we study the DE and the DDE algorithms along with a hash-chain and then analyze their FPRs as a function of both $m$ and $\beta$.

\subsection{Hash-Chain Based Embedding and Verification}

Suppose the packet originates from a source node with node ID $n_{i_{1}}$, for $i_{1} \in N$. Along with the Bloom filter contents, this node will use a seed value $hc_0 = K_{seed}$ for creating the hash-chain. A globally-known bit sequence could be used as $hc_0$. Assuming the use of a standard hash function $H(.)$ (e.g., SHA-256) for deriving the hash-chain, each embedding node in the path, say node $i_{j}$, updates the hash-chain with the help of following attributes (i) ID of embedding node, denoted by $ID_{i_{j}}$ (this can be edge ID or double-edge ID depending on the embedding technique), (ii) previous hash-chain value $hc_{i_{j-1}}$ (or initial seed value if the packet is from originating node), and the (iii) packet sequence number $seq$. With that, the updated hash-chain value is given by $hc_{i_{j}} = H(seq||ID_{i_{j}}||hc_{i_{j-1}})$. This way, once the packet reaches the destination, it uses the Bloom filter contents and the received hash-chain to recover the provenance. Following a standard hash-chain verification protocol, we note that hash verification is not needed when a single path is recovered; therefore, the destination need not spend extra resources for verifying the hash-chain. Otherwise, the destination will check for the correct path until it either gets the correct path or has used up its $\beta$ chances to find the correct path. If the destination fails to obtain the correct path within $\beta$ attempts, then we refer it to as the false-positive event. By adopting this technique, we observe that the energy consumption at the destination remains within its capability. Overall, the use of hash-chain assisted embedding algorithms (as shown in Fig. \ref{fig:provenance_embedding_process}) helps in ultra-reliability; this is because perfect path recovery is possible with probability one as long as the destination is ready to verify the hash-chains of all the paths in the topology.

\begin{figure}
    \centering
    \includegraphics[height = 6.5cm , width = 15.5cm]{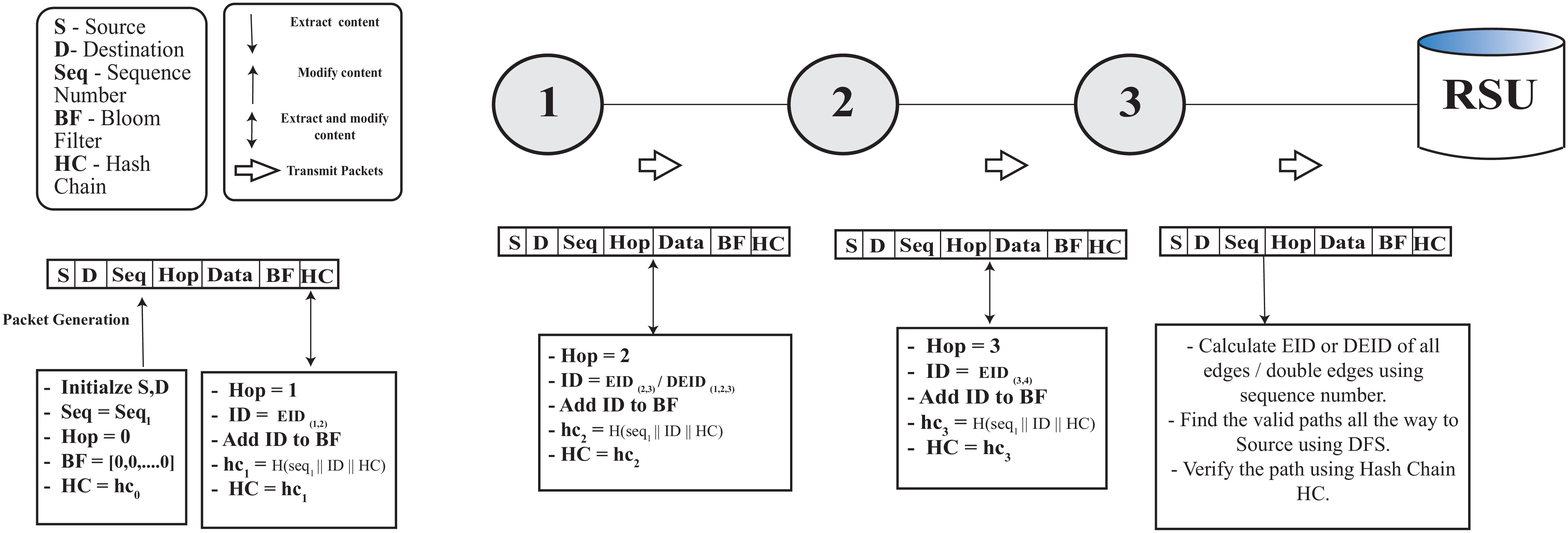}
    \caption{Depiction of the embedding process in payload phase for a three-hop network. In the first-hop under edge embedding, node 1 embeds the edge $(1,2)$ in the Bloom filter, and also generates the hash-chain $\mbox{\textbf{hc}}_{1}$. Similarly, the process of adding identities to the Bloom filter and generating hash-chain are shown in the boxes for each hop.}
    \label{fig:provenance_embedding_process}
\end{figure}

\subsection{Edge Embedding}

In the deterministic edge embedding (DE) method \cite{BFP}, each node embeds the directed edge ID associated with the next node in the chain. As a result, the DE method assists the RSU in recovering the provenance even if the underlying topology contains cycles. In addition, the accuracy of provenance recovery can be further improved owing to the use of hash-chain verification. To execute hash-chain verification along with the DE method, each node uses the same directed edge ID for updating the hash-chain.

\subsubsection{Embedding and Recovery Process}

Suppose that the packet travels the ordered sequence of nodes, denoted by node $i_{1}$, node $i_{2}$, $\ldots$, node $i_{h}$. After following the initialization steps, during packet transmission, node $i_{j}$ with its ID $n_{i_{j}}$, extracts the hop-counter and the Bloom filter $\mathbf{BF} = [BF[0],BF[1],\ldots,BF[m-1]]$ from the preceding node. During packet transmission, node $i_{j}$ generates an edge ID with node $i_{j+1}$, denoted by $EID_{(i_{j},i_{j+1})}$ (as discussed in Section \ref{subsec:SSMP:embedding}), and is then  embedded into $\mathbf{BF}$ as $BF[v_{i_{j}}^{(r)}] = 1$, where $v_{i_{j}}^{(r)} = H_m(EID_{(i_{j},i_{j+1})} || seq || r)$, for $1\leq r \leq k$, such that $H_m(\cdot)$ is a hash function which generates a number in $[0,m-1]$. Once the Bloom filter is updated, node $i_{j}$ also updates the hash-chain as $hc_{i_{j}}=H(seq||EID_{(i_{j},i_{j+1})}||hc_{{i_{j-1}}})$, where $hc_{{i_{j-1}}}$ is the value of hash-chain embedded by the preceding node. Finally, the hop-counter value is incremented by one and the packet is forwarded to the next node.

Meanwhile, at the RSU, provenance recovery process is performed by checking the membership of all the directed edges in the topology. Once the set of edges has been recovered from the Bloom filter, the recovery algorithm will perform a depth-first search (DFS) algorithm for finding valid paths from a given source node of hop length $h$. If the given path has the same hash-chain value as that received from the packet, then the DFS returns the path successfully without proceeding further. Otherwise, the DFS algorithm will backtrack and explore more paths. The DFS algorithm terminates if it either gets the correct path or has visited $\beta$ paths while performing the search. As a consequence, a false-positive event is said to occur if the correct path is not recovered after verifying $\beta$ paths.

 





 \subsection{ Double-Edge Embedding}
 
In this section, we present a variant of the DDE method, which is known to reduce the delay in provenance embedding when compared to the DE method \cite{BFP}. To facilitate low-latency, only alternate relay nodes on the path, embed the information of two edges connected to it; the edge through which the packet is received and the edge over which the packet is forwarded. In addition to these known advantages, our variant further improves its accuracy owing to the topology knowledge as well as the use of hash-chain. 
  
\subsubsection{Embedding and Recovery Process}

Similar to the DE method, we suppose that the packet travels the ordered sequence of nodes, denoted by node $i_{1}$, node $i_{2}$, $\ldots$, node $i_{h}$. During packet transmission, node $i_{j}$ extracts the hop-counter value from the packet. Since provenance embedding is performed by alternate nodes, node $i_{j}$ updates the Bloom filter and the hash-chain depending on whether hop-counter value is even or odd \cite{BFP}. If node $i_{j}$ has to update the provenance, it generates a double-edge ID with node $i_{j-1}$ and node $i_{j+1}$, denoted by $DEID_{(i_{j-1}, i_{j},i_{j+1})}$.\footnote{The procedure to generate $DEID_{(i_{j-1}, i_{j},i_{j+1})}$ can be similar to generating $EID_{(i_{j},i_{j+1})}$, wherein the identity of the double-edge can be synthesized as $Enc_{K_{i}}(n_{i_{j-1}} || n_{i_{j}} || n_{i_{j+1}})$.} Subsequently, the rest of the embedding process as described in the DE method will be followed. However, the only exception is that instead of using the edge ID $EID_{i_{j}, i_{j+1}}$, the double-edge ID $DEID_{(i_{j-1}, i_{j},i_{j+1})}$ will be used to update the Bloom filter and the hash-chain. Finally, the hop-counter is incremented by one before forwarding the packet to the next node.

Upon receiving the packet, the RSU uses all the double-edge IDs of the topology to verify their membership in the Bloom filter. Note that the set of double-edges recovered from the Bloom filter may contain more than $\lceil \frac{h}{2} \rceil$ double-edges owing to hash collisions. Similar to the DE technique, the destination recovers the path traced by the packet using a DFS algorithm on the recovered double-edges. As a consequence, we define a false-positive event, in which despite checking all the $\beta$ paths using the hash-chain, the RSU fails to recover the provenance from the packet.

\subsection{Optimization of Bloom filter parameters}

Having defined the false-positive events for the DE and DDE methods, we derive analytical expressions on their FPRs so that the expressions can be used to pick the right choice of $m$ and $k$ for a given affordable complexity at the destination (quantified by $\beta$). In contrast to \cite{BFP}, our approach uses topology knowledge as well as hash-chains when deriving the expressions.

\subsubsection{False-Positive Rates of DE and DDE Methods}

From the DE method, we define its FPR as $P_{\mathcal{E}}(fp) =\frac{N_{fail}}{N_{total}}
$, where $N_{fail}$ denotes the total number of times the RSU is unable to recover the provenance despite getting $\beta$ chances to verify the hash-chains, and $N_{total}$ is the total number packets transmitted. Similarly, the FPR of the DDE method, denoted by $P_{\mathcal{DE}}(fp)$, can be defined. In the following theorem, we derive analytical expressions for upper bounds on $P_{\mathcal{E}}(fp)$ and $P_{\mathcal{DE}}(fp)$ as a function of $\beta$, $k$ and $m$. 

\begin{theorem}
\label{thm_DEE_PFA}
Given $G(N,E)$, $m$, $k$, and $\beta$, an upper bound on FPRs of the DE and DDE methods can be obtained in closed-form.
\end{theorem}
\begin{proof}
We refer the reader to Proof \ref{prf_DEE_PFA} in the appendix section.
\end{proof}

\subsection{Comparison of Embedding Techniques}

The main objectives of this section are (i) to compare the upper bounds on the FPRs (from Theorem \ref{thm_DEE_PFA}) with that of the simulation results, and (ii) to present the improvements in the FPRs when compared to that of no topology knowledge during provenance recovery \cite{BFP}. To generate the simulation results, we use a network of $n = 20$ nodes with two topologies: one with 54 edges and the other with 34 edges. A comparison between the proposed upper bounds on the DE and the DDE methods and their simulation results are presented in Fig. \ref{fig:edge_embedding_sim_analytical} and Fig. \ref{fig:double_edge_embedding_sim_analytical}, respectively, for different values of $\beta$. Both plots confirm that the FPRs improve with increased complexity capability at the destination, and moreover, the value of $k$ for which the proposed closed-form expressions achieve the minima is close to that provided by simulation results. As a result, these expressions can be used to choose the right values of $k$ for a given $m$ and $\beta$.

We also present simulation results to analyze the advantage of learning the topology for the provenance recovery process. In particular, we compare the FPRs of the DE and the DDE methods with and without the topology knowledge. To generate the results, we use a network of $n = 20$ nodes with both 54 edges and 34 edges. We use $\beta = 1$, thereby not using hash-chains to recover the provenance. The plots, which are presented in Fig. \ref{fig:compare_knowledge}, confirm that with the knowledge of the topology, the FPRs improves substantially, and moreover, the order of improvement increases as the sparsity increases.

\begin{figure}
    \centering
    \includegraphics[trim={0 0 0 6cm},clip,scale=0.2]{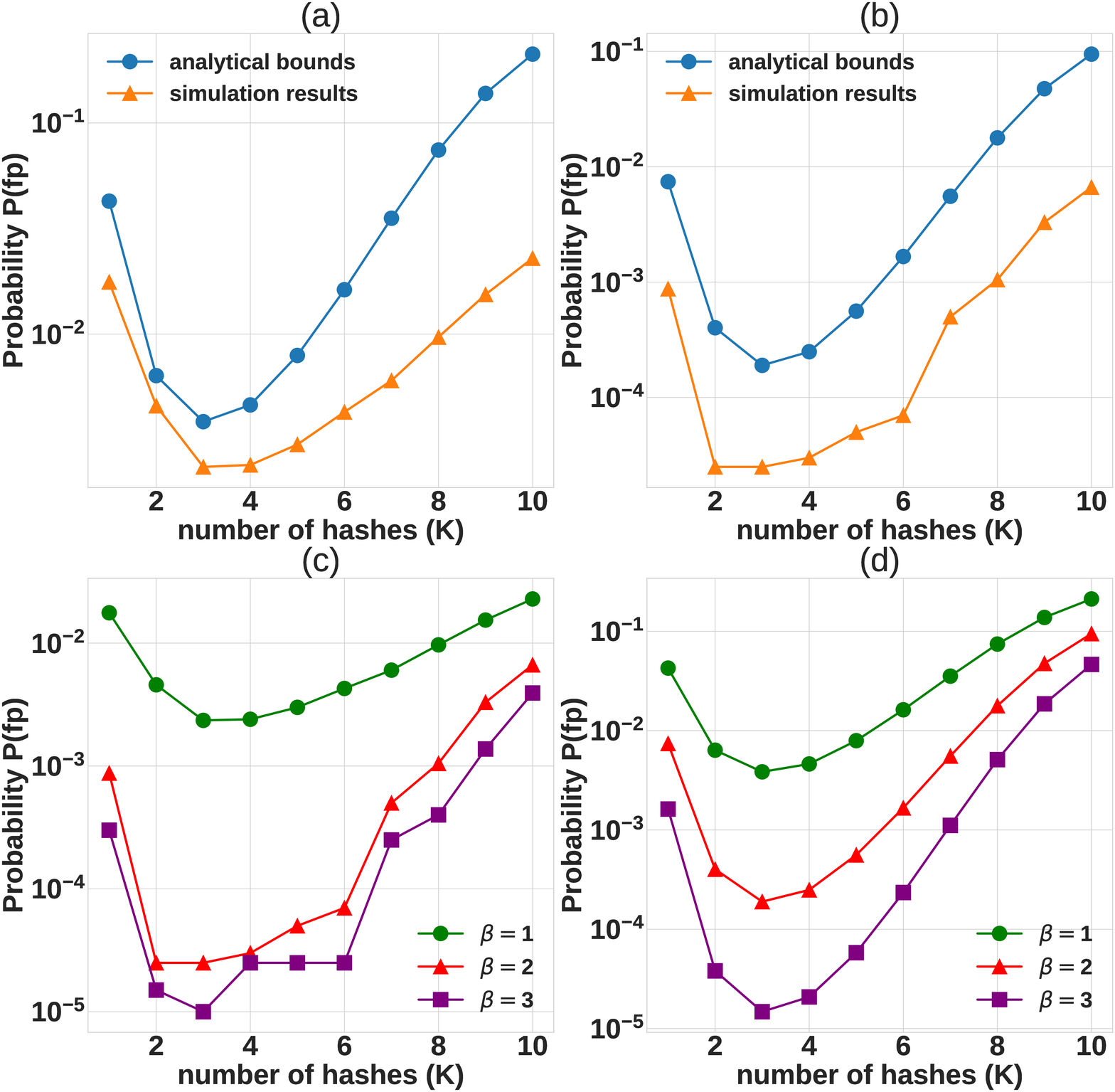}
    \vspace{-1cm}
    \caption{(a) Comparison between the analytical bound and the simulation results of DE with $\beta=1$ and $m = 20$ on a network with $n = 20$. (b) Setting as as in (a) with $\beta=2$. (c) Simulation results on FPRs with $\beta=1,2,3$ for $m = 20$ bits. (d) Analytical bounds on FPRs for $\beta=1,2,3$ for $m = 20$ bits.}
    \label{fig:edge_embedding_sim_analytical}
\end{figure}

\begin{figure}
    \centering
    \includegraphics[trim={0 0 0 6cm},clip,scale=0.2]{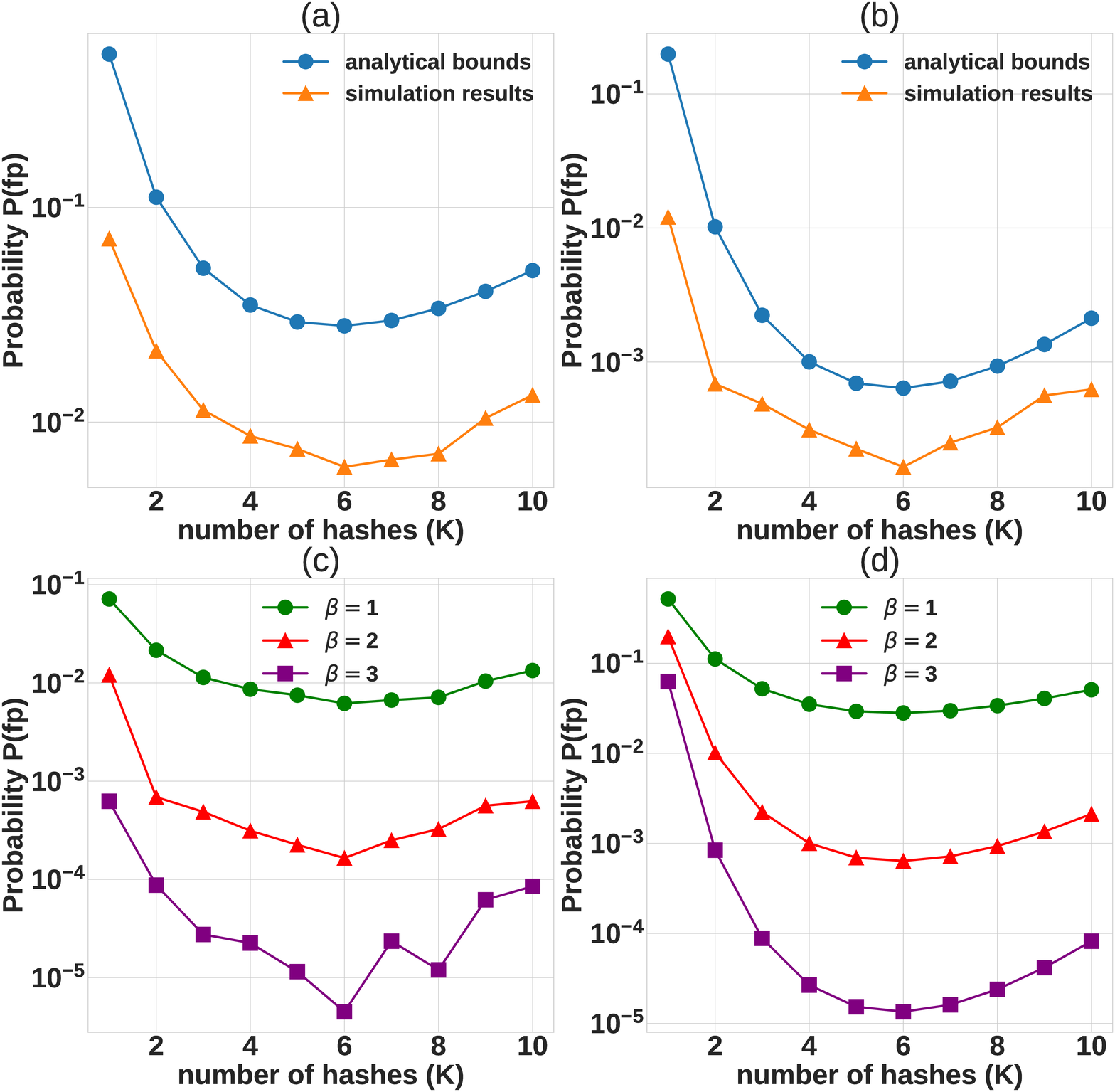}
    \vspace{-1.3cm}
    \caption{(a) Comparison between the analytical bound and the simulation results of DDE with $\beta=1$ and $m=20$ on a network of $n = 20$ nodes. (b) Setting same as in (a), however, with $\beta=2$. (c) Simulation results of FPRs for $\beta=1,2,3$ for $10^5$ packets with $m = 20$. (d) Analytical bounds on FPRs for $\beta=1,2,3$ with $m = 20$.}
    \label{fig:double_edge_embedding_sim_analytical}
\end{figure}

\begin{figure}
    \centering
    \includegraphics[trim={0 0 0 6cm},clip,scale=0.2]{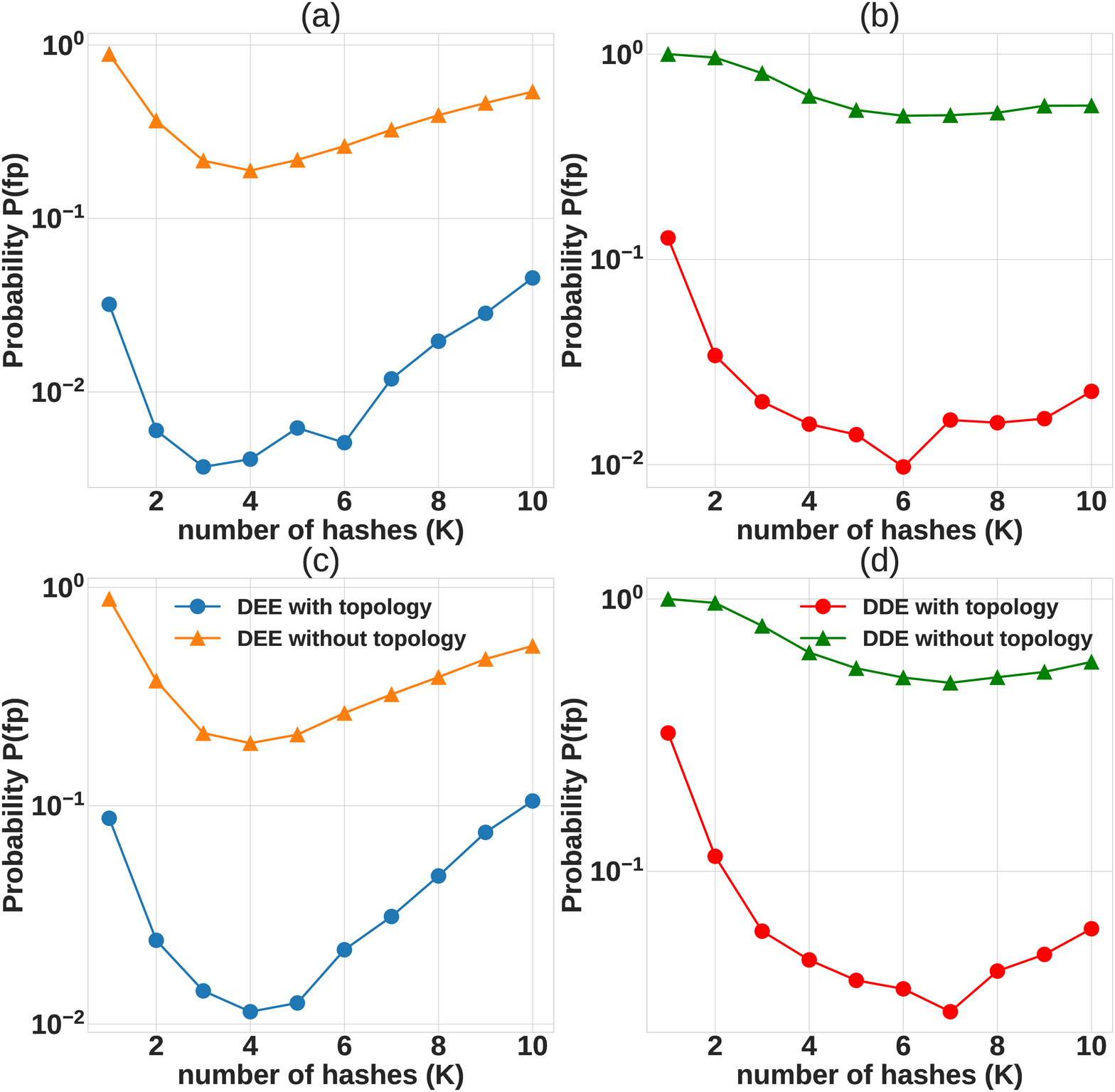}
    \vspace{-1.2cm}
    \caption{(a) Simulation results on DE with and without topology knowledge using $\beta=1$ and $m = 20$ on a network of $n = 20$ nodes containing 34 edges. (b) DDE method on the same setting as in (a). (c) Simulation results on DE with a graph containing $n = 20$ nodes and $54$ edges. (d) DDE on the same setting as in (c).}
    \label{fig:compare_knowledge}
\end{figure}

\section{Complexity Analysis and Simulation Results on OMNeT++}

In this section, first, we present a thorough analysis on the computation/communication overheads of the proposed schemes (both topology learning phase and payload phase). Subsequently, we implement the proposed schemes on OMNeT++ to measure these overheads in terms of the time taken to execute their functionality.


\subsection{Complexity Analysis}

With respect to the topology learning phase, the time complexity of our protocols can be analysed based on the following parameters: \textbf{(i)} Processing at each node, \textbf{(ii)} Queuing, forwarding and routing in network, and \textbf{(iii)} Processing at RSU for learning the topology. Therefore, by including the time components for various operations (as defined in Table \ref{tab:time_def_different_operations}), we compute the overall time taken by the proposed protocols.

\subsubsection{SSMP and MSSP Protocols}

From the definition of the SSMP protocol, the time taken for embedding the neighbour information at node $i$ will be  $t_{h, N}k\gamma_i$ (or $O(k\gamma_i)$). With $h_i$ denoting the number of hops taken by the packet from node $i$ to reach the RSU, the average time taken for propagation is given by $h_i(t_{q,N} + t_{pr})$ seconds (or $O(h_i)$). Subsequently, the RSU will check the membership of all the edges with node $i$, which will take $t_{q,R}+t_{h,R}n$ seconds. Finally, a mutual confirmation will be triggered when the packets from all the nodes have been processed by the RSU. This confirmation will take $O(n^2)$ comparison operations on the adjacency matrix (assuming all the nodes start the operation of embedding the neighbours simultaneously). Thus, for the packet of node $i$, we have: (i) Processing time at the embedding node is $t_{h,N}k\gamma_i$ (or $O(k\gamma_i)$), (ii) Queuing, forwarding and routing time is $h_i(t_{q,N}+t_{pr})$ (or $O(h_i)$), and finally, (iii) Processing time at RSU is $t_{q,R}+t_{h,R}n$. Therefore the overall time taken by the packet of $n_i$ is $t_{h,N}k\gamma_i+h_i(t_{q,N}+t_{pr})+t_{q,R}+t_{h,R}nk$. Although this expression provides the time taken for registering the edges of one node, the total time taken by all the nodes is less than their sum since the SSMP protocol allows parallelism both during packet routing as well as the edge verification process. This aspect will be clear when we present the simulation results using OMNeT++ environment.

For the MSSP protocol, the time taken at each node for embedding their neighbours is $t_{h,N}k\gamma_i$ (or $O(k\gamma_i)$). Also, there will be no queuing required for any node because only one packet will be traversing all the nodes. Assuming the maximum number of edges through which the packet passes before reaching the RSU to be $h_{max}$, the propagation time will take $h_{max} t_{pr}$. Meanwhile, at the RSU, the total number of hash calculations is $2 {{n} \choose {2}}k$, which approximately takes $2 {{n} \choose {2}}kt_{h,R}$ seconds when executed sequentially. Thus, (i) the processing time at the nodes is $t_{h,N}k\gamma_i$ (or $O(k\gamma_i)$), (ii) forwarding and routing time is $h_{max}t_{pr}$ (or $O(n)$), and finally, (iii) processing time at RSU is $t_{h,R}(2{n \choose{2}}k)$. Therefore, the total time taken by the MSSP protocol is:
$t_{h,N}k\sum_{i=1}^{n-1}\gamma_i+h_{max}t_{pr}+t_{h,R}(2{n \choose{2}}k) 
$.

\begin{table}
\caption{Time taken for different operations during the protocol}

    \centering
    \begin{tabular}{|c|c|}
        \hline
         \textbf{Notation}& \textbf{Description} \\
         \hline
         $t_{h,N}$& Time taken for calculating the hash at nodes.\\
         \hline
         $t_{h,R}$ & Time taken for calculating the hash at RSU. \\
         \hline
         $t_{q,N}$ & Average time taken for processing in queue at node\\
         \hline
         $t_{q,R}$ & Average time taken for processing in queue at RSU\\
         \hline
         $t_{pr}$ & Propagation time \\
         \hline
    \end{tabular}
    \label{tab:time_def_different_operations}
\end{table}

\subsubsection{Provenance Embedding during Payload Phase}
In the payload phase, the total time taken for packet transmission depends upon the embedding technique used for adding the details of the path through which the packet has traversed. Additionally, since a hash is being calculated for hash-chain embedding, the total number of hash calculations will be one more than the number of hashes used to embed the Bloom filter. Therefore, the total time taken for packet transmission for the edge embedding and double-edge embedding schemes are $h((k+1)t_{h,N}+t_{p})$, and $\lfloor\frac{h}{2}\rfloor((k+1)t_{h,N}+t_{p})$, respectively. Since the delay benefits of the DDE method is already known during packet routing \cite{BFP}, we only focus on the delay benefits of the topology knowledge during the provenance recovery process.

With respect to the provenance recovery process, the recovery time depends on the embedding method. In particular, the number of hash calculations to recover the edges and the double-edges at the RSU will be $(2{n \choose{2}}-(n-1))k$ and  $(6{n\choose{3}}-2(n-1)(n-2))k$, respectively. Hence, the time taken for processing the hash calculations for edge embedding is $(2{n \choose{2}}-(n-1))kt_{h,R}$ seconds, and for double-edge embedding, it will be $(6{n\choose{3}}-2(n-1)(n-2))kt_{h,R}$ seconds. Finally, once the required information on the edges and double-edges are recovered, the DFS algorithm is invoked. In particular, the DFS algorithm will terminate as soon as the hash of the path of hop-length $h$ is matched with the received hash. In the worst case, if the received hash value is not matched, then DFS algorithm will run only for $\beta +1$ times. Therefore, the total time taken for hash-chain computation will be less than  $\beta h t_{h,R}$ seconds for edge embedding, and $\beta \lfloor \frac{h}{2} \rfloor t_{h,R}$ seconds for double-edge embedding. Thus, the overall time taken for edge and double-edge embedding will be $h((k+1)t_{h,N}+t_{p})+(2{n \choose{2}}-(n-1))kt_{h,R}+\beta h t_{h,R}$ seconds and $\lfloor\frac{h}{2}\rfloor( (k+1)t_{h,N}+t_{p})+(6{n\choose{3}}-2(n-1)(n-2))kt_{h(R)}+\beta\lfloor\frac{h}{2}\rfloor t_{h,R}$ seconds, respectively.

For the case when the topology knowledge is known at the RSU, the transmission time for the edge and double-edge embedding methods will remain the same as that without the topology knowledge at the RSU. However, the number of hashes that are required for the verification of the edges (or the double-edges) will be $(|E|)k$ (or $(|DE|)k$), where $E$ and $DE$ are the set of edges and the double edges in the topology. Thus, in the case of sparse topology, the order of $E$ and $DE$ are small, i.e., $|E|<<(2{n \choose{2}}-(n-1))$ and $|DE|<<(6{n\choose{3}}-2(n-1)(n-2))$, thereby significantly reducing the time for recovering the edges/double edges.

\subsection{Simulation Results on OMNeT++}
\begin{figure}
    \centering
    \includegraphics[trim={0.1cm 0.2cm 0.1cm 0.1cm},clip,scale=0.35]{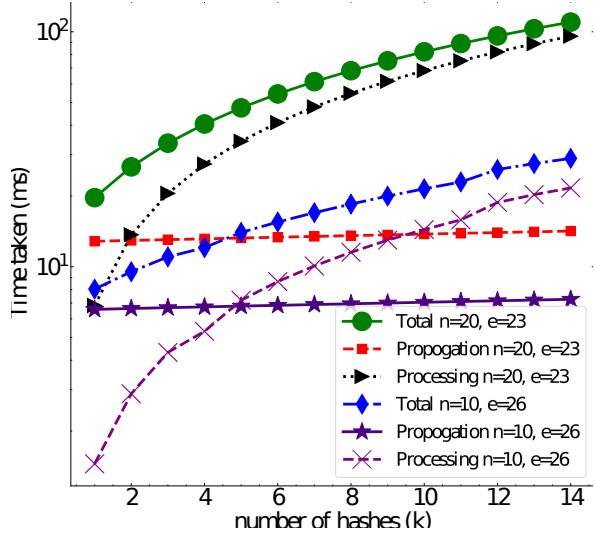}
    \caption{Time taken by various parts of the MSSP Protocol when implemented on OMNeT++.}
    \label{fig:MSSP_comb}
\end{figure}
In this section, we present the communication/computation overheads of the proposed schemes after implementing them on the OMNeT++ environment. To generate the results, simulations were carried out on \textbf{OMNeT++}  simulator (version \textit{5.6.2}) with system configuration having Intel Core-i5 processor and 8 GB RAM.  We used the INET framework for the simulation of an ad-hoc wireless network that uses AODV routing.  The simulations involved (i) a topology with $n = 10$ nodes and $e = 26$ edges, (ii) a topology with $n = 20$ and  $e = 23$ edges, and (iii) a topology with $n = 20$ and $e = 34$ edges. Furthermore, for the simulations, we have considered a transmission delay of $0.5 ms$ between two nodes (assuming the use of a packet equivalent to one slot in LTE standard), processing time for computing hash function once as $42 \mu s$ and $10 \mu s$ at nodes and RSU, respectively (this is to capture difference in processing complexity at the nodes and the RSU). Moreover, only for the SSMP protocol, we have considered a queue processing time as $70 \mu s$, as in SSMP, queue is required because a node may be processing a packet when the packet is received from other nodes. 

\begin{figure}

    \centering
    \includegraphics[trim={0.1cm 0.2cm 0.1cm 0.1cm},clip,scale=0.35]{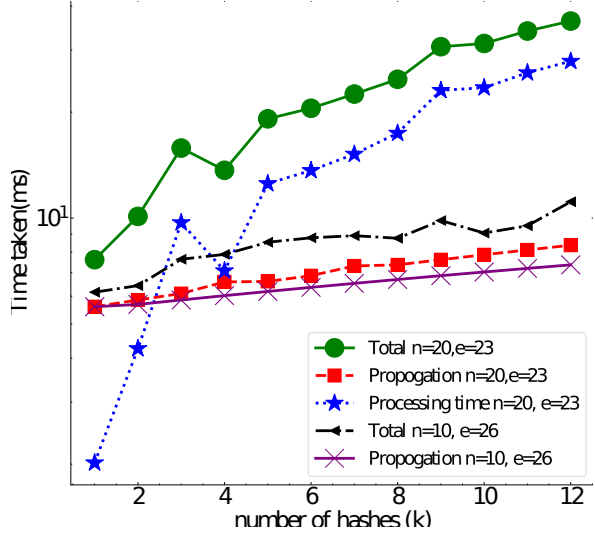}
    \caption{Time taken by various parts of the SSMP Protocol when implemented on OMNeT++.}
    \label{fig:SSMP_comb}
\end{figure}

Using the delay parameters discussed above, we plot the time taken by various blocks of the MSSP and the SSMP algorithm in Fig. \ref{fig:MSSP_comb} and Fig. \ref{fig:SSMP_comb}, respectively. In particular, the following terms are used in the plots: (i) Total time - it is the sum of propagation time, processing time at each node, and the processing time at RSU to learn the topology, (ii) Processing time - it only includes the time taken by the RSU to learn the topology, and (iii) Propagation time - it is the total time taken by all the nodes in embedding the packet and forwarding it to the next node. From Fig. \ref{fig:MSSP_comb} (for the MSSP protocol) and Fig. \ref{fig:SSMP_comb} (for the SSMP protocol), we can see that the total time taken by the SSMP algorithm is lower than that of the MSSP algorithm, and this observation is attributed to the fact that the former algorithm is able to achieve parallelism as RSU can verify the edges as and when packets arrive, and moreover, the packets from multiple sources can simultaneously flow through the network. We also observe that dominant portion of the total time is the time taken by the processing operation at the RSU. Moreover, when comparing the propagation delay between the two protocols in Fig. \ref{fig:Comp_prop}, we observe that for the MSSP protocol, the propagation delay decreases with the increasing network density (increase in the number of edges), and this reduction is because of fewer number of extra transmissions in order to reach unvisited nodes. We also observe that the overall propagation delay of the SSMP protocol is lower than the MSSP protocol owing to parallelism achieved by simultaneous flow of packets from multiple sources. 

\begin{figure}

\centering
\includegraphics[trim={0.1cm 0.2cm 0.1cm 1cm},clip,scale=0.35]{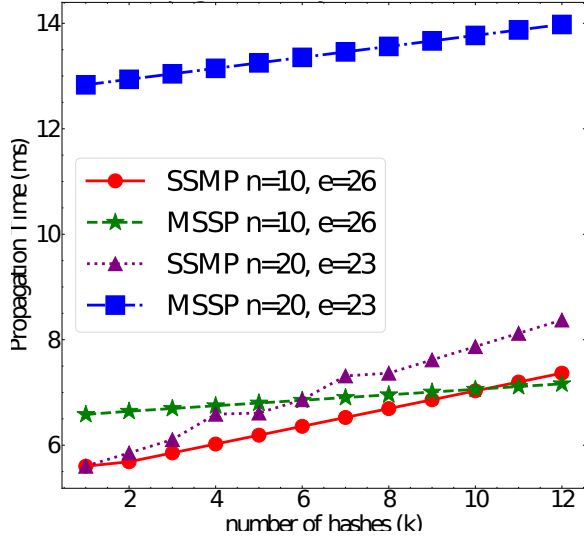}
\caption{Comparison of propagation delay between the SSMP and the MSSP algorithms}
\label{fig:Comp_prop}
\end{figure}

Finally, in Fig. \ref{fig:withAndw_oDEE}, we compare the delay incurred for executing provenance recovery with and without using the knowledge of the topology. As far as the benefits of the topology knowledge is concerned, we observe that fewer the number of edges (double-edges) in the topology, lower is the time taken for the provenance verification process at the RSU. These simulation results confirm the advantages of our protocols over the contributions of \cite{BFP}.




In conclusion, the above OMNeT++ based simulation results suggest that the SSMP algorithm must be preferred because of its low end-to-end delay for the topology learning phase. Furthermore, the results also suggest that a sparse topology must be learned and then used in the payload phase in order to further improve the FPRs and the end-to-end delay during the provenance recovery phase.

\section{Conclusion}
Identifying the challenges of latency constraints and changing topology in sparse vehicular networks, we have proposed novel strategies for ultra-reliable topology learning and provenance recovery. Besides presenting a rigorous analysis on the FPRs of the topology learning phase and the payload phase, we have also demonstrated our ideas by implementing them on OMNeT++ environment to confirm the latency benefits. Since we address the joint problem of topology learning and its subsequent use in provenance recovery, we believe that the following tasks: (i) collection of number of neighbours from each node, (ii) solving Problem \ref{problem:2}, and then (iii) executing the topology learning phase, must be completed within a small fraction of the coherence time so that the rest of the fraction can be used for the payload phase and provenance recovery. As far as the decision of whether to learn the topology is concerned,  the RSU must receive the information on the number of neighbours of each node, and then determine whether the topology is dense or sparse, and then execute the topology learning phase if the network is sparse. Although this work does not address questions related to the choice of the threshold, we advocate to keep the threshold on the number of edges in the network as small as possible since an extremely-sparse topology incurs low communication-overhead for the topology learning phase, and also provides significant benefits in time complexity and FPRs for the payload phase.

In this work, we have used the first part of the coherence time to learn the topology of the network with no prior knowledge of the network. However, in practice, once the first batch of topology learning is completed, we expect the topology to vary over time in a gradual manner. As a result, an interesting future work is to develop secure \emph{topology-tracking} algorithms. 

\begin{figure}
    \centering
    \includegraphics[trim={0.1cm 0.2cm 0.1cm 0.1cm},clip,scale=0.35]{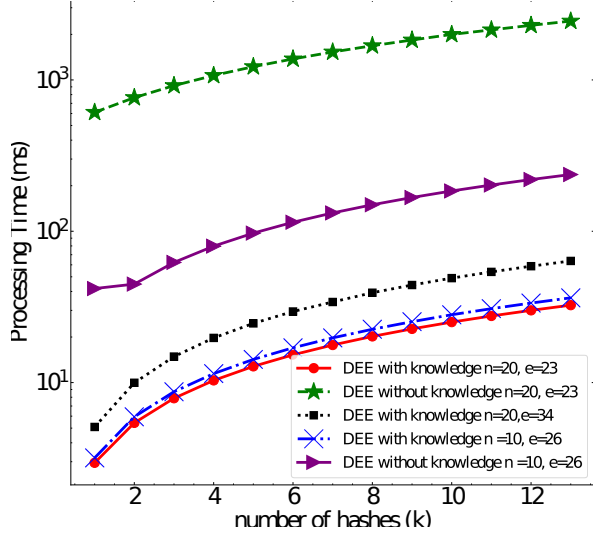}
    \caption{Plots depicting the reduction in processing delay for the provenance recovery process when using the knowledge of the topology.}
    \label{fig:withAndw_oDEE}
\end{figure}




\section*{Acknowledgments}

This work was supported by the Indigenous 5G Test Bed project from the Department of Telecommunications, Ministry of Communications, New Delhi, India.

\section{Appendix}

\begin{proof1}
\label{prf:SSMP_exact}
   [\textbf{Proof for Theorem \ref{thm:SSMP_exact}}] We note that \textbf{1)} The false-positive events are caused due to the occurrence of edges that are not in the topology. \textbf{2)} The edges that are either originating from or terminating at the destination are not considered for the false-positive event as the destination already has the knowledge of its neighbours through the neighbour-discovery process. To capture the above points when defining the false-positive events, let us consider the complementary graph of the given graph and also remove from it the complementary edges that are connected to the destination. Let the set of all edges in the complete graph be denoted by $E_c$, defined as $E_c=\{(i,j)~|~\forall i, j \in N \mbox{ such that } i \neq j\}$. Also, let the set $E_d$ denote the set of edges that are connected to the destination. With that the complementary graph excluding the edges of destination is given by $\overline{G}(N,\overline{E})$, where $\overline{E} \triangleq E_c-E_d-E$ such that $-$ represents the set difference operator. For a false-positive event to occur, at least one edge from $\overline{E}$ must be recovered during the topology learning phase. Formally, for two nodes $x, y \in N$, let $E_{x \leftrightarrow y}$ denote the event when the edge $(x, y) \in \overline{E}$ is recovered in the learning phase. Based on the recovery algorithm, the event $E_{x \leftrightarrow y}$ is defined as
    \begin{equation}\label{eqn:2}
     E_{x \leftrightarrow y}=E_{x \rightarrow y} \land E_{y \rightarrow x},
  \end{equation}
 where $E_{x \rightarrow y}$ and $E_{y \rightarrow x}$ are the events that edge $(x,y)$ and edge $(y,x)$ have been recovered from the Bloom filter shared by node $x$ and node $y$, respectively. Thus, considering the recovery of any edge in $\overline{E}$, the false-positive event is written as $fp = \bigcup_{(a,b) \in \overline{E}} E_{a \leftrightarrow b}.$ As a consequence, the FPR of the SSMP scheme is
     $P(fp)= P\biggl(\bigcup_{(a,b) \in \overline{E}} E_{a \leftrightarrow b}\biggr).$ Let us denote $r \triangleq |\overline{E}|$, and also denote $\overline{E}$ as $\{e_1,e_2,e_3,\ldots,e_r\}$, where $e_i$ denotes the $i$-{th} edge in the set when enumerated in some fashion. Since the probability of union of multiple events can be written using inclusion-exclusion principle, we have 
 \begin{eqnarray*}
     P\left( \bigcup_{1 \leq i \leq r}  e_i\right) = \sum_{1 \leq {i_1} \leq r} P\left( e_{i_1}\right)
- \sum_{1 \leq {i_1} < {i_2} \leq r} P\left( e_{i_1} \cap e_{i_2}\right)
 + \sum_{1 \leq {i_1} < {i_2} < {i_3} \leq r}
P\left( e_{i_1} \cap e_{i_2}\cap e_{i_3}\right) \\
- \ldots
+ (-1)^{r+1} P\left( \bigcap_{i=1}^r e_i \right).
 \end{eqnarray*}
In the expression for $P\left( \bigcup_{1 \leq i \leq r}  e_i\right)$, we notice that $P \left( e_1 \cap e_2 \cap e_3 \cap ... e_z\right)= \prod_{i=1}^{z}P(e_i),$ for any $2 \leq z \leq |\overline{E}|$, and this is because the generation of the index values in the Bloom filter follow an identical and statistically independent process. As a result, the exact expression for $P\left( \bigcup_{1 \leq i \leq r}  e_i\right)$ can be computed by computing $P \left( e_i \right)$ for $1 \leq i \leq r$. Towards that direction, in the rest of the proof, we compute the expression for $P(E_{x \leftrightarrow y}) = P(E_{x \rightarrow y}\land E_{y \rightarrow x})$. Again, since the embedding process at each node is statistically independent, we can write $P(E_{x \leftrightarrow y})= P(E_{x \rightarrow y})P( E_{y \rightarrow x}),$ and therefore, we only focus on the expression for $P(E_{x \rightarrow y})$. From first principles, $P(E_{x \rightarrow y})$ is given by
\begin{equation} \label{eqn:5}
P(E_{x \rightarrow y})=\sum_{i=1}^{min(m_x,k_x\gamma_x)} P(E_{x \rightarrow y}| S_i)P(S_i), 
\end{equation}
where $S_i$ is the event that exactly $i$ bits of the Bloom filter are set in the packet sent by node $x$, and $P(E_{x \rightarrow y}| S_i)$ is the conditional probability that the $k_{x}$ locations chosen for the edge $(x, y)$ coincides with the $i$ locations of the Bloom filter. Towards computing $P(S_{i})$, we need to compute the following attributes: (i) The total number of ways in which $k_{x}\gamma_x$ index positions can be set out of $m_{x}$ distinct positions, which is given by $A = m_x^{k_x\gamma_x}$, (ii) The total number of ways in which we can choose $i$ distinct positions out of $m_x$ positions, given by $B =$ $ {m_x}\choose{i}$, and finally, (iii) The total number of ways in which we need to select $k_x\gamma_x$ indices out of $i$ given indices such that each index in $i$ is selected at least once, which in turn can be solved as $C = S(\gamma_x k_{x},i)(i!),$ where $S(.,.)$ is the Stirling number of the second kind, defined as
    $$S(k,n)=\frac{1}{n!}\sum_{i=0}^{n}(-1)^{i} {{n}\choose{i}}(n-i)^k. \vspace{-0.2cm}$$ Thus, the overall expression for $P(S_{i})$ can be written as $\frac{BC}{A}$, expanded as,
    
 \begin{equation} \label{eqn:6}
     P(S_i)=\frac{{m_x \choose i}\sum_{j=0}^{i}(-1)^{j} {{i}\choose{j}}(i-j)^{\gamma_x k_x}}{m_x^{\gamma_x k_x}}.
 \end{equation}
 On the similar lines, the expression for $P(E_{x \rightarrow y}| S_{i})$ is
 \begin{equation} \label{eqn:7}
 P(E_{x \rightarrow y}| S_{i})={ \left( \frac{i}{m_x} \right)} ^ {k_x}, 
 \end{equation}
 which captures the probability that the hash function outputs for the edge $(x, y)$ pick up the $i$ bits that are already set in the Bloom filter. Substituting \eqref{eqn:6} and \eqref{eqn:7} in \eqref{eqn:5}, we obtain
 \begin{equation*}
      P(E_{x \rightarrow y})= \sum_{i=1}^{min(k_x \gamma_x,m_x)}\frac{i^{k_x}{ {m_x} \choose i}\sum_{j=0}^{i}(-1)^{j} {{i}\choose{j}}(i-j)^{\gamma_x k_x}}{m_x^{(\gamma_x+1)k_x}}.
 \end{equation*}
 Once the above type of expressions are computed for each edge in $\overline{E}$, we can compute $P\left( \bigcup_{1 \leq i \leq r} e_i\right)$ in closed-form. 
 \end{proof1}
 
\begin{proof1}
\label{prf:SSMP_bound}
[\textbf{Proof for Theorem \ref{thm:SSMP_bound}}] It is well known that the FPRs can be upper bounded using the union bound as 
 \begin{equation*}
     P\biggl(\bigcup_{(x,y) \in \overline{E}} E_{x \leftrightarrow y}\biggr) < \sum_{(x,y) \in \overline{E}} P\left( E_{x \leftrightarrow y} \right),
 \end{equation*}
 where $E_{x \leftrightarrow y}$ is the event when the edge $(x, y)$ is recovered in the topology learning phase. We know that $P(E_{x \rightarrow y})$ is a function of $m_x,\gamma_x$ and $k_x$, and likewise, $P(E_{y \rightarrow x})$ is a function of $m_y,\gamma_y$ and $k_y$. Therefore, writing $P(E_{x \rightarrow y})=f_{fp}(m_x,\gamma_x,k_x)$, we rewrite the union bound as
 \begin{equation}
 \label{eq:semi_upper_bound}
     \sum_{(x,y) \in \overline{E}} f_{fp}(m_x,\gamma_x,k_x) f_{fp}(m_y,\gamma_y,k_y).
 \end{equation}
 Out of the $|\overline{E}|$ terms in the above expression, we know that $f_{fp}(m_x,\gamma_x,k_x)$, which is the term corresponding to node $x$, appears $n-\gamma_{x}-2$ times. However, since the topology information is not known, we do not know its counterpart terms, which are of the form $f_{fp}(m_y,\gamma_y,k_y)$, i.e., the nodes that would be connected to node $x$ in the complementary graph. To circumvent this problem, we will proceed to compute an upper bound on \eqref{eq:semi_upper_bound} by assuming that in the complementary graph, node $x$ is connected to those nodes that have a large number of neighbours. In other words, we will artificially connect the false-positive edges to those $n-\gamma_x -2$ nodes which are most likely to occur. Formally, let us sort the $n-1$ nodes as $N_{sorted}=\{s1, s2, s3, s4, \ldots, s(n-1)\}$, wherein the sorting is done based on the evaluation of $f_{fp}$ on the parameters of each node as $f_{fp}(m_{si},\gamma_{si},k_{si}) \leq f_{fp}(m_{sj},\gamma_{sj},k_{sj})$ for $i < j$. Note that this is possible since the sets $\{\gamma_1,\gamma_{2},\ldots,\gamma_{n-1}\}$, $\{m_1,m_2,m_3, \ldots, m_{n-1}\}$, and $\{k_1,k_2,k_3, \ldots, k_{n-1}\}$ are fixed. For each node $sx$ in the sorted list, we pick the last $n-\gamma_{sx}-2$ distinct nodes (excluding node $sx$) of $N_{sorted}$, and use the corresponding $f_{fp}$ values when computing the counterparts of node $sx$ in \eqref{eq:semi_upper_bound}. By denoting this set of $n-\gamma_{sx}-2$ nodes by $S_{x}$, an upper bound on the FPRs can be written as
\begin{equation}
\label{eq:upper_bound}
P'(fp)=\sum_{x \in N_{sorted}} \sum_{y \in S_x} f_{fp}(m_{x},\gamma_{x},k_{x})f_{fp}(m_{y},\gamma_{y},k_{y}).
\end{equation}
This completes the proof. 
\end{proof1} 

     \begin{proof1}
     \label{prf:MSSP}
     [\textbf{Proof for Theorem \ref{thm:MSSP}}] Similar to the proof on FPRs of the SSMP protocol, we consider the set $\overline{E}$, which comprises the set of edges not present in the actual topology, excluding the complementary edges connected to the destination. Since the total number of edges in the topology is $|E|$, the total number of edges embedded on the packet will be $2|E| - \gamma_{RSU}$, where the factor $2$ captures the fact that a given edge is embedded twice by both its vertices, and $\gamma_{RSU}$, which represents the number of neighbours of the destination, is discounted as the RSU does not embed its neighbours in the Bloom filter. From first principles, the FPR of the MSSP protocol is given by $P(fp)=\sum_{i=1}^{|\overline{E}|}P(fp,i),$ where $P(fp,i)$ denotes the probability of the false-positive event wherein $i$ non-existing edges, for $1 \leq i \leq |\overline{E}|$, are recovered in the topology learning process. Furthermore, since the false-positive event also depends on the number of bits already set in the Bloom filter, we can write 
     \begin{equation*}
        P(fp,i)=\sum_{j=1}^{min(m,(2|E|-\gamma_{\bl{RSU}})k)}P((fp,i) | S_j )P(S_j),  
    \end{equation*}
    where $P(S_j)$, given by
    \begin{equation*}
     P(S_j)=\frac{{m \choose j}\sum_{t=0}^{j}(-1)^{t} {{j}\choose{t}}(j-t)^{\gamma k}}{m^{\gamma k}},
 \end{equation*}
   denotes the probability that $j$ bits are set in the Bloom filter such that $\gamma = 2|E|-\gamma_{RSU}$, and $P((fp,i) | S_j )$ denotes the probability that $i$ non-existing edges appear in the Bloom filter conditioned on the event $S_{j}$. Given that the false-positive events of each non-existing edge are statistically independent, we can write $P((fp,i) | S_j )$ using binomial expansion as 
     \begin{equation*}
P((fp,i) | S_j )={{|\overline{E}|}\choose{i}}\left(\delta\right)^i\left(1-\delta\right)^{|\overline{E}|-i},
\end{equation*}
where $\delta=\left( \frac{j}{m} \right)^{2k}$ is the probability that the $2k$ index values chosen by both the vertices of a non-existing edge lies on the $j$ indices already set in the Bloom filter. By plugging all the derived equations, the overall FPR is given by 

\begin{small}
\begin{equation*}
P(fp)=\sum_{i=1}^{|\overline{E}|}\sum_{j=1}^{min(m,(2|E|-\gamma_{RSU})k)}{{|\overline{E}|}\choose{i}}\left(\delta\right)^i\left(1-\delta\right)^{|\overline{E}|-i}P(S_j).
\end{equation*}
\end{small}
\end{proof1}

\begin{proof1}
\label{prf:security_imp_attack}
[\textbf{Proof for Proposition \ref{prop:security_imp_attack}}] \bl{Suppose that node $d$, for some $d \in N$, attempts to add the edge $(x, y)$, for some $x \neq d, y$, in the Bloom filter. Since node $d$ does not have the identity of the edge $(x, y)$ (since it is private and unclonable), it attempts to randomly generate $k$ statistically independent index values in the Bloom filter with uniform distribution. In such a case, the probability of success of impersonation attack is the probability that the index values generated by node $d$ coincides with that of node $x$ when it would embed the edge $(x, y)$. We observe that the success-rate of impersonation attack depends on the number of index values chosen by the other edges in the network. Formally, the success-rate is 
\begin{equation}
\label{eq:success-rate_imp_attack}
p_{succ} = \sum_{i = 1}^{min(m,(2|E|-\gamma_{RSU})k)} \mbox{Pr}(C_{i}) \sum_{j = 0}^{k} \binom {m - i}{j} (p_{match, j})^{2},
\end{equation}
where $\mbox{Pr}(C_{i})$ is the probability that $i$ positions, for $1 \leq i \leq min(m,(2|E|-\gamma_{RSU} )k)$, of the Bloom filter are chosen by all the legitimate edges, the term $p_{match, j}$ is the probability that node $d$ chooses $k$ index values in the Bloom filter such that $j$ distinct index values, for $0 \leq j \leq k$, are chosen outside the set of $i$ index values (which are already chosen by the other edges) and the remaining $k - j$ index values are chosen at any of those $i + j$ index values of the Bloom filter. Note that the term $p^{2}_{match, j}$ appears in \eqref{eq:success-rate_imp_attack} due to statistical independence between the index values chosen by node $d$ and node $x$. We also show that $p_{match, j}$ is lower bounded by $\frac{\binom {k}{j}((j!)(i^{(k-j)}))}{m^{k}}$ using the inclusion-exclusion principle  Therefore, $p_{succ}$ can be lower bounded by
\begin{equation}
\label{eq:success-rate_impersonation}
\sum_{i} \mbox{Pr}(C_{i}) \sum_{j = 0}^{k} \binom {m - i}{j} \left(\frac{\binom {k}{j}((j!)(i^{(k-j)}))}{m^{k}}\right)^{2},
\end{equation}
where the range of $i$ is same as in \eqref{eq:success-rate_imp_attack}.
}
\end{proof1}

\begin{proof1}
\label{prf_DEE_PFA}
[\textbf{Proof for Theorem \ref{thm_DEE_PFA}}] Let $G(N,E)$ be such that there exists a total of $\lambda$ distinct paths of hop-length $h$ from a given source to the destination. Since the destination is capable of verifying at most $\beta$ paths, a false-positive event can occur when more than $\beta$ paths are recovered from the Bloom filter. Conditioned on a given path travelled by the packet, let the corresponding set of edges (or double-edges) be represented by $E_{actual}$ (or $DE_{actual}$). Excluding the path chosen by the packet, a false-positive event can occur if at least $\beta$ paths out of the remaining $\lambda -1$ paths appear in the Bloom filter. To count such events, there are ${\lambda-1}\choose{\beta}$ distinct ways denoted by $\{C_1,C_2,\ldots,C_{{\lambda-1}\choose{\beta}}\}$, and for each combination the $\beta$ distinct paths are represented by $\{P_{i1},P_{i2},P_{i3},\ldots,P_{i\beta}\}$,
for $1 \leq i \leq$ $ {\lambda-1}\choose{\beta}$. With this, using the union bound, we can write $P_{\mathcal{W}}(fp) < \sum_{i = 1}^{{\lambda-1}\choose{\beta}} P(X[C_{i}])$, where $\mathcal{W} \in \{\mathcal{E}, \mathcal{DE}\}$, $P(X[C_{i}])$ is the probability of occurrence of $\beta$ paths in $C_{i}$, captured by the event $X[C_{i}]$. Towards computing $P(X[C_{i}])$, we have ${X[C_{i}]=\bigcap_{j=1}^{\beta} X[P_{ij}]}$, where $X[P_{ij}]$ denotes the event wherein the edges/double-edges of the path $P_{ij} = n^{(1)}_{ij} \rightarrow n^{(2)}_{ij} \rightarrow \ldots n^{(h-1)}_{ij} \rightarrow n^{(h)}_{ij}$ are recovered from the Bloom filter. The corresponding set of edges and double-edges are $\{(n^{(1)}_{ij},n^{(2)}_{ij}),(n^{(2)}_{ij},n^{(3)}_{ij}), \ldots \}$ and $\{(n^{(1)}_{ij},n^{(2)}_{ij}, n^{(3)}_{ij}),(n^{(3)}_{ij},n^{(4)}_{ij}, n^{(5)}_{ij}), \ldots \}$. With such edges/double-edges, there are two possibilities: \textbf{Case 1}: Some edges/double-edges of $P_{ij}$ may belong to the path travelled by the packet, and as a result, those edges/double-edges will always be recovered with probability one. \textbf{Case 2}: Some edge/double-edges of $P_{ij}$ do not belong to the path travelled by the packet, and as a result, the probability of recovering such an edge/double-edge is
 \begin{equation*}
p = \sum_{i=1}^{min(k \theta,m)}\frac{i^{k}{ {m} \choose i}\sum_{j=0}^{i}(-1)^{j} {{i}\choose{j}}(i-j)^{\theta k}}{m^{(\theta+1)k}},
\end{equation*}
wherein we substitute $\theta = h$ and $\theta = \lfloor{\frac{h}{2}}\rfloor$ for the DE and DDE methods, respectively.
With the DE method, let $E_{ij}$ represent the set of distinct edges of the path $P_{ij}$ that are not present in $E_{actual}$. Similarly, with the DDE method, let $DE_{ij}$ represent the set of distinct double-edges of the path $P_{ij}$ that are not present in $DE_{actual}$. Thus, by considering all such edges/double-edges under the $X[C_{i}]$, we get $E_i=\bigcup_{j=1}^{\beta}E_{ij}
$ and $DE_i=\bigcup_{j=1}^{\beta}DE_{ij}$. Since the event of hash-collision of each edge/double-edge are identical and statistically independent, we can write $P(X[C_i])=p^{|E_i|}$ and $P(X[C_i])=p^{|DE_i|}$ for the DE and the DDE method, respectively. Hence, an upper bound on the false positive rate for the DE method is given by $\overline{P}_{\mathcal{E}}(fp)= \sum_{i=1}^{{\lambda-1}\choose{\beta}}P(X[C_i])$,
which is equal to $\sum_{i=1}^{{\lambda-1}\choose{\beta}}p^{|E_i|}
$. Similarly, in the case of DDE method, an upper bound on the false positive rate is given by
$\overline{P}_{\mathcal{DE}}(fp)=\sum_{i=1}^{{\lambda-1}\choose{\beta}}p^{|DE_i|}
$. This completes the proof.
\end{proof1}

\bibliographystyle{elsarticle-num}
\bibliography{refs}




\end{document}